\documentclass[aps,pra,groupedaddress,amssymb,amsfonts,showpacs, nofootinbib]{revtex4}
\usepackage{verbatim}
\usepackage{amsmath, bbold} 
\usepackage{hyperref}   
\usepackage{graphicx} 
\usepackage{subfigure} 

\usepackage{url}

\usepackage{fancyheadings} 

\usepackage{frcursive}
\usepackage[T1]{fontenc}

\newcommand{\mathfrc}[1]{\text{\fontfamily{frc}\footnotesize\selectfont#1}}

\usepackage{ mathrsfs } 

\def\curG{{\mathcal{G}}}
\def\curS{{\mathscr{S}}}
\def\curg{{\mathfrc{g}}}
\def\curs{{\mathfrc{s}}}
\def\curh{{\mathfrc{h}}}
\def\cure{{\mathfrc{e}}}
\def\curL{{\mathfrc{L}}}
\def\calL{{\mathcal{L}}}

\newcommand{\cH}{{\cal H}}

\def\unitaryset{{\mathcal{U}}}

\newcommand{\leavenode}[2]{#1\bullet\!\!\!\xrightarrow{#2}}

\usepackage{amsthm} 
\makeatletter
\def\th@plain{  \thm@notefont{}  \itshape}
\def\th@definition{ \thm@notefont{}  \normalfont}
\makeatother

\theoremstyle{definition}
\newtheorem{definition}{Definition}
\newtheorem{protocol}{Protocol}
\newtheorem{example}{Example}

\theoremstyle{theorem}
\newtheorem{theorem}{Theorem}
\newtheorem{lemma}{Lemma}

\theoremstyle{definition}
\newtheorem{remark}{Remark}

\DeclareMathOperator{\id}{\mathbb{1}} 

\def\C{{\mathbb{C}}}
\def\Z{{\mathbb{Z}}}

\def\N{{\mathbb{N}}}
\def\F{{\mathbb{F}}}

\newcommand{\tr}[0]{{\rm tr}}

\newcommand{\loc}{\ell} 
\newcommand{\Hamindex}{k} 
\newcommand{\Sqk}[1][q]{S_{#1}^{(k)}}
\newcommand{\rhot}{\rho^{\otimes \loc}} 
\newcommand{\dist}{\delta}

\renewcommand*{\leq}{\leqslant}
\renewcommand*{\geq}{\geqslant}

\makeatletter
  \newcommand{\miniscule}{\@setfontsize\miniscule{9}{10}}%

\usepackage{mathtools}
\newcommand\SmallMatrix[1]{{%
  \tiny\arraycolsep=0.3\arraycolsep\ensuremath{\begin{pmatrix}#1\end{pmatrix}}}}
  
\newcommand\MiniMatrix[1]{{%
	\SmallMatrix{#1}}}

\usepackage[usenames, dvipsnames]{color} 

\begin{document}

\title{Improved bounded-strength decoupling schemes for local Hamiltonians\par}

\author{Adam D. Bookatz} 
\affiliation{Center for Theoretical Physics, Massachusetts Institute
of Technology, Cambridge, MA, U.S.A}

\author{Martin Roetteler}
\affiliation{Microsoft Research, Quantum Architectures and Computation Group, Redmond, WA, U.S.A.}

\author{Pawel Wocjan}
\affiliation{Department of Electrical Engineering 
and Computer Science, University of Central Florida, Orlando, FL, U.S.A \\} 


\begin{abstract}
We address the task of switching off the Hamiltonian of a system by removing all internal and system-environment couplings. We propose dynamical decoupling schemes, that use only bounded-strength controls, for quantum many-body systems with local system Hamiltonians and local environmental couplings. To do so, we introduce the combinatorial concept of balanced-cycle orthogonal arrays (BOAs) and show how to construct them from classical error-correcting codes. The derived decoupling schemes may be useful as a primitive for more complex schemes, e.g., for Hamiltonian simulation. For the case of $n$ qubits and a $2$-local Hamiltonian, the length of the resulting decoupling scheme scales as $O(n \log n)$, improving over the previously best-known schemes that scaled quadratically with $n$. More generally, using balanced-cycle orthogonal arrays constructed from families of BCH codes, we show that bounded-strength decoupling for any $\ell$-local Hamiltonian, where $\ell \geq 2$, can be achieved using decoupling schemes of length at most $O(n^{\ell-1} \log n)$. 
\end{abstract}

\pacs{{03.67.Lx, 03.65.Fd, 03.67.-a}}

\maketitle
\thispagestyle{fancy} 

\section{Introduction}
Consider a quantum system of $n$ interacting $d$-dimensional qudits with a time-independent (possibly unknown) Hamiltonian $H$ acting on a Hilbert space $\cH\cong (\C^d)^{\otimes n}$. 
We make the assumption that the system is $\loc$-local, i.e. that $H$ can be written as the sum of operators, each of which acts only on $\loc$ of the $n$ qudits. In nature it is usually the case that $\loc$ is small even when $n$ is large.
Without loss of generality, we also take $H$ to be traceless, and for technical reasons, we assume that $d$ is a prime power (which includes the important case of qubits, i.e. $d=2$).

We consider the task of \emph{decoupling}, i.e. effectively switching off the Hamiltonian $H$ (including removing any couplings to the environment) so that the system effectively evolves under the zero Hamiltonian. Such a task is important, for example, in the context of quantum memory, where one desires to preserve the state of a quantum system.

To achieve this task, we assume that the natural dynamics of the system can be modified by adjoining an
open-loop (non-feedback) controller according to
\begin{equation*}
H \mapsto H + H_c(t) \,.
\end{equation*}
In practice, physical limitations restrict the types of control Hamiltonians available for use. We consider the realistic setting in which $H_c(t)$ is only 1-local, i.e. due to our limited control of the system, $H_c$ is the sum of operators that each act on only one qudit.  We further impose the constraint that our control Hamiltonian $H_c(t)$ is limited to be bounded-strength, i.e. a sufficiently smooth bounded function.  This is in contrast to the setting of bang-bang control in which $H_c(t)$ can be a discontinuous function that takes values of arbitrarily large norm.
Our assumptions that the system Hamiltonian is an $\loc$-local Hamiltonian acting on a system of $n$ interacting qudits and that the control Hamiltonian is a $1$-local bounded-strength Hamiltonian reflect the typical composite nature of quantum systems and their coupling locality as well as the limitations in implementing external controls.

Viola and Knill proposed a general method for bounded-strength decoupling; see \cite{VK:02} and \cite[Chapter 4]{QECbook}.   Their method, often referred to as \textit{Eulerian decoupling}, relies on Eulerian cycles in Cayley graphs of a control group --- a certain finite group of control unitaries that can be implemented by switching on control Hamiltonians, from a finite set of available control operations, for a fixed time.  The Eulerian cycle dictates which control Hamiltonians are applied in the different time-slots of the decoupling protocol.    

The Eulerian method, as introduced in \cite{VK:02}, does not make it possible to directly leverage the fact that the system Hamiltonian is $\loc$-local in order to obtain more efficient decoupling schemes.  However, in the setting of bang-bang control there do exist efficient decoupling schemes that are specifically designed for composite quantum systems with $\loc$-local system Hamiltonians; see \cite{stollsteimer,finite,equivalence} and \cite[Chapter 15]{QECbook}. In these schemes, the specification of which bang-bang control unitaries are to be applied is chosen according to the entries of so-called \emph{orthogonal arrays of strength $\loc$}. They are matrices with the property that any submatrix formed by an arbitrary collection of $\loc$ rows satisfies a certain balancedness  condition.  

The work \cite{eulerOA} presented a particular construction of decoupling schemes merging the approaches of Eulerian (bounded-strength) decoupling together with orthogonal array (bang-bang) decoupling.  This construction yields schemes that require only bounded-strength controls and exploit the composite structure of the quantum system (namely, the locality of the system Hamiltonian) to achieve decoupling with fewer control operations.  To do so, these schemes introduce the concept of so-called \textit{Eulerian orthogonal arrays}.

The purpose of the present paper is to further improve upon the method of \cite{eulerOA} to obtain even more efficient bounded-strength decoupling schemes.  To this end, we first generalize the Eulerian method due to \cite{VK:02} by showing that it is also possible to achieve decoupling with the help of so-called \textit{balanced cycles}, which encompass Eulerian cycles as a special case.  We then show that bounded-strength decoupling of composite quantum systems with local Hamiltonians can be accomplished based on the new concept of \textit{balanced-cycle orthogonal arrays}.

Note that all the schemes discussed above can also be applied to the situation of a general open quantum system with joint Hamiltonian $H$ acting on a quantum system that is coupled to an uncontrollable environment. Such a Hamiltonian has the form
\begin{equation*}
H = H_S \otimes \id_B + \id_S \otimes H_B + \sum_\alpha S_\alpha
\otimes B_\alpha ,
\label{openHam}
\end{equation*} 
where the operators $H_S$ and $S_\alpha$ act on the system
and where the operators $H_B$ and $B_\alpha$ act on the environment. We assume that the system Hamiltonian $H_S$ and the operators $S_\alpha$ are all $\loc$-local. The decoupling goal in this case is to effectively switch off the system Hamiltonian $H_S$ and remove all couplings to the environment.
If, using controls that act only on the system, one can effectively switch off all generic system Hamiltonians, then such an operation will switch off $H_S$ and each $S_\alpha$, thereby accomplishing decoupling.\footnote{The remaining Hamiltonian term of $\id_S \otimes H_B$ is inconsequential, as it does not affect the system at all.} For notational simplicity, the remainder of the paper will therefore ignore the environment and treat only the case of effectively switching off an arbitrary $\loc$-local operator $H$.

%
\section{Description of the control-theoretic model}\label{Sec:control_model}
Consider the group $(\F_q,+)$, the additive group of the finite field of order $q=d^2$, where $d$ (the dimension of the qudits) is some prime power.
For the remainder of this paper, let $\rho :\F_q \rightarrow \unitaryset(d)$ be a faithful, irreducible, unitary, projective%
\footnote{Projective representations need only be homomorphisms up to phase, i.e. obey $U_{g + h} \propto U_g U_h$ with proportionality rather than equality.} %
representation 
that maps the elements of $\F_q$ to $d\times d$ unitary matrices, say $\rho : g \mapsto U_g$.
That $q$ cannot be smaller than $d^2$ for such a representation will be justified later in Remark~\ref{remark:sizes}; that $q=d^2$ suffices is justified by the explicit example shown below.

We assume that for every $g\in\F_q$ we can implement $U_g$ on any qudit of our system in the following sense: for every $g$, we can physically implement, over time $\delta\in[0,\Delta]$, a bounded-strength single-qudit Hamiltonian $h_g(\delta)$,
corresponding to a single-qudit unitary evolution operator $u_g(\delta)$, such that $U_g = u_g(\Delta)$ where $\Delta$ is some fixed length of time. We assume that we can do this on any qudit and, moreover, that we can do so for each of the $n$ qudits in parallel. Note that this assumption obeys the practical control limitations discussed earlier.

Of particular interest, in the case of qubits ($d=2, q=4$) we can consider the representation $\rho :\F_4 \rightarrow \unitaryset(2)$ that maps the four elements of $\F_4$ to the four $2\times 2$ Pauli matrices $\lbrace \id, X, Y, Z \rbrace$. Thus, it is assumed that we can physically implement any Pauli operator on any qubit. Rather than assuming that $q=4$, this paper will treat $q$ more generally; however the reader is invited to think of the special case of qubits if desired.
For non-qubits, with $q>4$, we can generalize this example as follows. For a prime $p$, define $\tilde X = \sum_{j=0}^{p-1} |j+1 \text{ mod } p\rangle\langle j|$ and $\tilde Z=\sum_{j=0}^{p-1} \omega^j |j\rangle\langle j|$, where $\omega$ is a $p^\text{th}$ root of unity. For prime $d=p$, the map $(a,b) \mapsto \tilde X^a \tilde Z^b$ defines a faithful, irreducible, unitary, projective representation from $\Z_d \times \Z_d$ to $\unitaryset(d)$. For a prime power $d=p^e$ (for some $e$), map $((a_1,b_1),\ldots,(a_e,b_e)) \mapsto \tilde X^{a_1} \tilde Z^{b_1} \otimes \cdots \otimes \tilde X^{a_e} \tilde Z^{b_e}$.

A decoupling protocol is defined by specifying a sequence of control Hamiltonians (equivalently, control unitaries) to be applied. As shown in Fig. \ref{fig:array}, we construct an $n\times N$ array with entries from $\F_q$, which we regard as a sequence of $N$ columns from $\F_q^n$. The $j$th column $\vec g_j = (g_{1j},\ldots,g_{nj})^T$ corresponds to the $j$th time interval $\big[(j-1)\Delta, j\Delta \big]$ of our protocol, during which we apply the control Hamiltonian 
\[
h_{\vec g_j}(\delta) = h_{g_{1j}}(\delta)\otimes\id\otimes\cdots\otimes\id \quad + \cdots + \quad \id\otimes\cdots\otimes\id\otimes h_{g_{nj}}(\delta)
\]
that gives rise to evolution $u_{\vec g_j}(\delta) = u_{g_{1j}}(\delta)\otimes\cdots\otimes u_{g_{nj}}(\delta)$ over $\delta\in[0,\Delta]$. In other words, 
 for each $\delta\in [0,\Delta]$ and $j=1,\ldots,N$,
$H_c(t) = h_{\vec g_j}(\delta)$ where $t = (j-1)\Delta + \delta$. The total time required to apply the entire sequence, i.e. the control cycle length, is therefore $T_c = N\Delta$, at which point the control sequence can be repeated. Observe that for any $t=(j-1)\Delta+\delta$, the unitary evolution $U_c(t)$ corresponding to the control Hamiltonian consequently satisfies $U_c(t) = u_{\vec g_j}(\delta) U_c\big((j-1)\Delta\big)$.
\begin{figure}
	\begin{center}
		\includegraphics[width=4in]{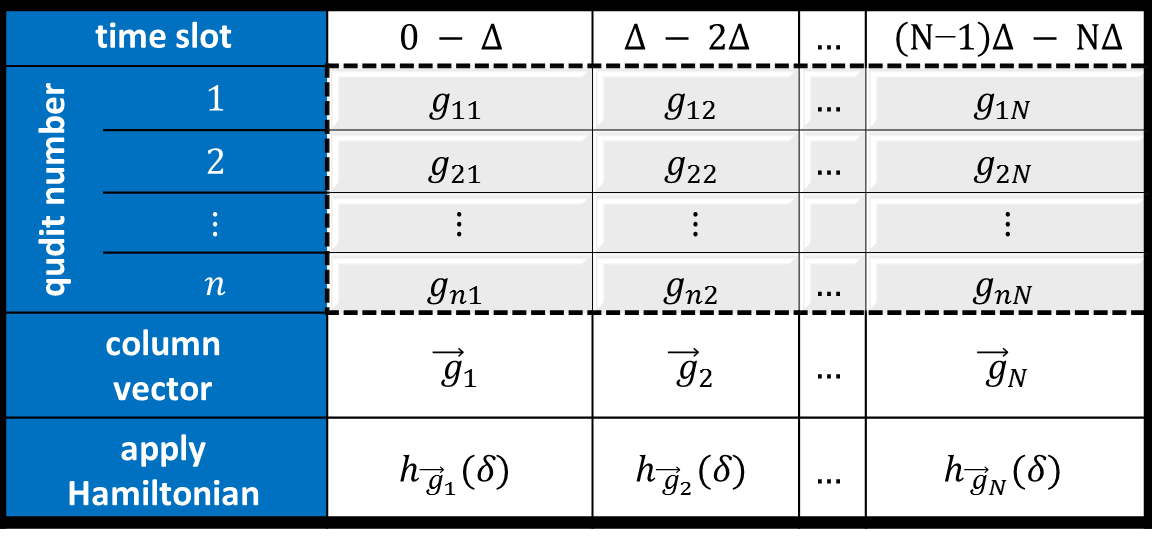}
		\caption
		{\label{fig:array}
		An $n\times N$ array, with each entry $g_{ij}\in\F_q$, shown within the dashed lines. Rows correspond to qudit numbers, columns to time slots (each of width $\Delta$). This array encapsulates the control sequence, with $H_c(t) = h_{\vec g_j}(\delta)$ over $\delta\in[0,\Delta)$ during the interval $t\in\big[(j-1)\Delta,j\Delta\big)$.
		}
	\end{center}
\end{figure}

According to average Hamiltonian theory \cite{EBW:87,WHH:68, Haeberlen76}, the resulting system evolution under $H+H_c(t)$ can be effectively approximated by 
\[
U(t) \approx e^{-i \bar H^{(0)} t}
\]
at times $t$ that are integer multiples of $T_c$, i.e. $t=mT_c$ for any $m\in\N$, where
\[
\bar H^{(0)} = \frac{1}{T_c} \int_{t=0}^{T_c} U_c(t)^\dag H U_c(t) dt
\]
is time-independent and where $U_c(t)$ is the time evolution due to $H_c(t)$ alone. The goal of decoupling, therefore, is to choose $U_c(t)$ such that $\bar H^{(0)} = 0$ for any $H$. It is in this sense that we effectively switch off the Hamiltonian $H$.
We refer the reader to \cite{eulerSim} for a detailed description of the above control-theoretic model and the resulting effective time-evolution.
We note, in particular, that although the approximation above is to leading order (in the Magnus expansion of $U_c(t)^\dag H U_c(t)$), the second-order term may be eliminated by designing the control Hamiltonian to satisfy $U_c(t)=U_c(T_c-t)$ \cite{Haeberlen76}.

The efficiency of the protocol developed in this paper is obtained by exploiting the composite structure of the Hamiltonian, namely the fact that $H$ was assumed to be a \textit{local} Hamiltonian. By definition, an $\loc$-local Hamiltonian $H$ on $n$ qudits can be written as $H = \sum_\Hamindex H_\Hamindex$, where each $H_\Hamindex$ acts non-trivially on at most $\loc$ of the $n$ qudits. In particular, the $\loc=2$ case corresponds to Hamiltonians with only pairwise interactions.
Our goal is to create a protocol that decouples each $H_\Hamindex$ simultaneously, and therefore decouples $H$. To see that this would work, observe that for any protocol $U_c(t)$,
\[
\bar H^{(0)} = \frac{1}{T_c} \int U_c(t)^\dag H U_c(t) dt = \sum_\Hamindex \frac{1}{T_c} \int U_c(t)^\dag H_\Hamindex U_c(t) dt = 
\sum_\Hamindex \bar H_\Hamindex^{(0)} .
\]

\section{Balanced cycles}
The success of the decoupling protocol introduced in this paper will rely on some basic group theory, which we introduce now.
Let $\curG$ be an Abelian group with a generating set $\curS\subset\curG$, i.e. any element of $\curG$ can be written as a sum of elements from $\curS$. 

\begin{definition}[Cayley graph]
The \textit{Cayley graph}, $\Gamma(\curG,\curS)$, of $\curG$ with respect to $\curS$ is a directed graph whose vertices are labeled by the group elements and whose edges are labeled by the generators. More precisely, there is a directed edge labeled $\curs$ from vertex $\curg\in\curG$ to vertex $\curh\in\curG$ iff $\curh = \curs+\curg$ for the generator $\curs\in\curS$.
\end{definition}

\begin{definition}[Cycle]
A \textit{cycle}, $\curL$, on $\Gamma(\curG,\curS)$ is a traversal on $\Gamma$ that starts and ends on the same vertex. We describe the cycle by the ordered list $\curL_\curG = \left(\curg_0,\ldots,\curg_{N-1}\right)$ of elements from $\curG$, indicating the order in which the elements are visited, with the understanding that the cycle visits $\curg_N = \curg_0$ immediately after visiting $\curg_{N-1}$. All the cycles in this paper visit every vertex at least once, so we assume without loss of generality that the first vertex is the identity element, $\cure$, of $\curG$. With this assumption we may equivalently represent the cycle $\curL_\curG$ by specifying the edges traversed, i.e. $\curL_\curS = (\curs_1,\ldots,\curs_N)$, where $\curg_{j}=\curs_{j}+\curg_{j-1}$ for $j=1,\ldots,N$; note that we differentiate between these representations by the subscript on $\curL$, but they both refer to the same cycle.
\end{definition}
Note that a cycle may visit vertices more than once and may traverse edges multiple times.
We will be interested not only in the vertices, but also the specific labels leaving each vertex; we denote by $\leavenode{\curg}{\curs}$ the $\curs$-labeled edge leaving vertex $\curg$.

\begin{definition}[Balanced cycle]
We say that $\curL$ is a \textit{balanced cycle} if 
$\forall \curs \in \curS, \exists \mu_\curs>0$ such that $\forall \curg\in\curG, \ \leavenode{\curg}{\curs}$ occurs exactly $\mu_\curs$ times; in other words, the cycle is balanced if it is balanced with respect to each label $\curs\in\curS$ in the sense that it leaves each $\curg$ via label $\curs$ an equal number of times (independent of $\curg$). Consequently, each $\curg$ will appear in $\curL$ precisely $\lambda = \sum_{\curs\in\curS} \mu_\curs$ times, independent of $\curg$. Because a Cayley graph is
a regular directed graph, it always has a balanced cycle whose length is then necessarily $N=\lambda |\curG|$.
\end{definition}

\begin{figure}[h]
	\begin{center}
	\subfigure[Eulerian cycle]{
		\includegraphics[width=2.5in]{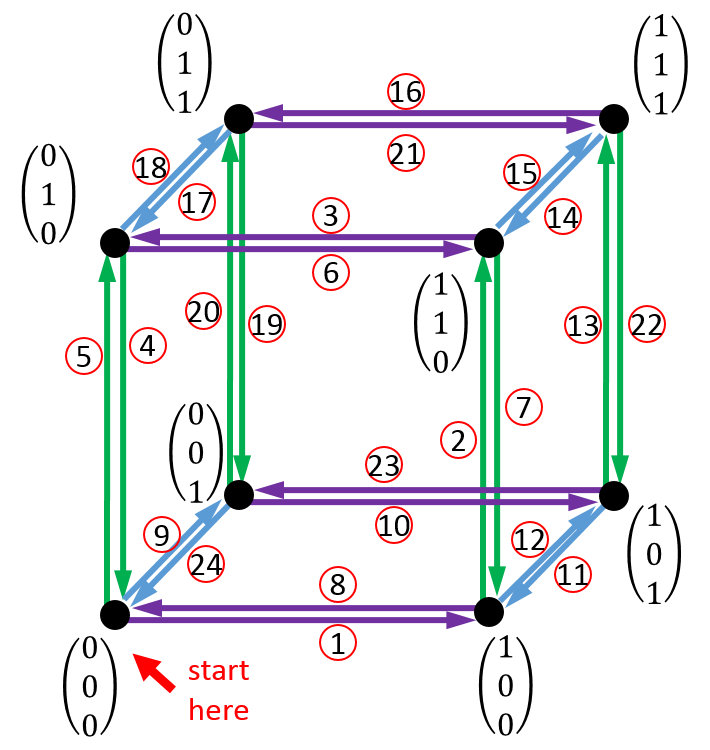}\label{fig:cycle_cube}
		}
	\subfigure[Balanced cycle]{
		\includegraphics[width=2.5in]{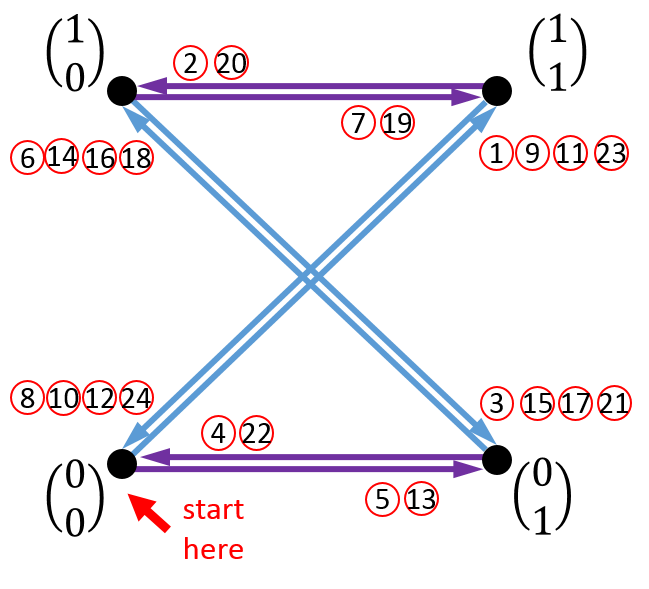}\label{fig:balanced_cycle_hourglass}
		}
		\caption[Examples of balanced cycles]
		{
		(a) An Eulerian cycle on the Cayley graph ${\small \Gamma\Big(\Z_2^3,\left\lbrace {\MiniMatrix{1\\0\\0}}, {\MiniMatrix{0\\1\\0}}, {\MiniMatrix{0\\0\\1}} \right\rbrace\Big)}$, i.e. a balanced cycle in which each edge label leaves each vertex precisely once. Vertices correspond to the eight elements of $\Z_2^3$. Edge labels correspond to the three generators, namely $\MiniMatrix{1\\0\\0}$ (purple), $\MiniMatrix{0\\1\\0}$ (green), and $\MiniMatrix{0\\0\\1}$ (blue). The cycle starts at $\MiniMatrix{0\\0\\0}$ and follows the path indicated (in ascending numerical order) by the circled integers (red).
		(b) A balanced cycle on the Cayley graph ${\small \Gamma\Big(\Z_2^2,\left\lbrace {\MiniMatrix{0\\1}}, {\MiniMatrix{1\\1}}\right\rbrace\Big)}$. Vertices correspond to the four elements of $\Z_2^2$. Edge labels correspond to the two generators, namely $\MiniMatrix{0\\1}$ (purple) and $\MiniMatrix{1\\1}$ (blue). The cycle starts at $\MiniMatrix{0\\0}$ and follows the path indicated (in ascending numerical order) by the circled integers (red). Observe that the cycle is indeed balanced: for each of the two edge labels, the edges leave each vertex the same number of times, irrespective of vertex. Specifically, the $\MiniMatrix{0\\1}$ label leaves each vertex precisely $\mu_{\MiniMatrix{0\\1}} = 2$ times, while the $\MiniMatrix{1\\1}$ label leaves each vertex precisely $\mu_{\MiniMatrix{1\\1}} = 4$ times.
		}
	\end{center}
\end{figure}
An important special case of a balanced cycle is an \textit{Eulerian cycle} on $\Gamma(\curG,\curS)$, for which $\mu_\curs=1$ for every $\curs\in\curS$. 
Examples of an Eulerian cycle and a non-Eulerian balanced cycle are shown in Fig.~\ref{fig:cycle_cube} and Fig.~\ref{fig:balanced_cycle_hourglass} respectively.
In \cite{VK:02}, Eulerian cycles were used to define decoupling protocols that avoided the discontinuous nature of bang-bang decoupling. More generally, one can define decoupling protocols based on balanced cycles (of which Eulerian decoupling is a special case), to which we now turn our attention.
Note, however, that this balanced-cycle decoupling protocol will not be the goal of this paper. Indeed, such a protocol will not exploit the composite structure of the Hamiltonian. Later we will utilize the balanced-cycle decoupling on $\loc$-qudit subsystems of a larger $n$ qudit space to develop more efficient protocols; in the current section, however, we may regard $\loc$ as the size of the entire system.

In exploiting the $\loc$-local nature of $H$, we will find that we are primarily interested in the group
\[
\curG = \F_q^\loc = \lbrace (a_1,\ldots,a_\loc)^T : a_i\in \F_q \rbrace
\]
with some generating set 
$\curS$
and the representation
\[
\rho^{\otimes \loc} : \curG \rightarrow \unitaryset(d^\loc)
\]
defined from our representation $\rho :\F_q \rightarrow \unitaryset(d)$. Specifically, if $\curg = (a_1,\ldots,a_\loc)^T \in\curG$ and $\rho(a_i)=U_{a_i}$ then $\rho^{\otimes \loc}(\curg) = U_\curg = U_{a_1}\otimes\cdots\otimes U_{a_\loc}$.
By our assumptions above, we can physically implement $U_\curg$ by applying the control unitary $u_\curg(\delta)$ (equivalently, the control Hamiltonian $h_\curg(\delta)$) for time $\Delta$.
For example, in the case of qubits ($q=4$), the group $\F_4$, whose elements we denote%
\footnote{Here the reader may prefer to equivalently think of the Abelian group as $\lbrace 1,z, x, y=xz=zx \rbrace$ with generating set $S_4 = \lbrace z,x \rbrace$. Then we can use $\curS = \lbrace z^{(1)}, x^{(1)}, \ldots z^{(\loc)}, x^{(\loc)} \rbrace$ and $\rho(x)=X, \rho(y)=Y,$ and $\rho(z)=Z$. Be aware, however, that the group operation used throughout the paper is denoted by $+$ rather than by multiplication, since it is inherited from the finite field.} %
as $\lbrace 0, 1, \alpha, \alpha+1 \rbrace$, is generated by the set $S_4 = \lbrace 1, \alpha \rbrace$. We choose $\curS = \lbrace 1^1, \alpha^1, \ldots 1^\loc, \alpha^\loc \rbrace$, which is a generating set of $2\loc$ elements for the group $\curG = \F_4^\loc$, where $x^i$ here denotes the column $(0,\ldots,0,x,0,\ldots,0)^T$ with $x\in\F_q$ in the $i$th position. In this case we assume $\rho(a_i)=U_{a_i}$ is a Pauli matrix, so $\rho^{\otimes \loc}(\curg)$ is a tensor product of Pauli matrices.

The purpose for the group theory used in this paper resides in the following observation \cite[Chapter 4]{QECbook}. 
We define the operator $\Pi_\curG$ to act on matrices $A$ as
\begin{equation}\label{eqn:Pi}
\Pi_\curG(A) = \frac{1}{|\curG|} \sum_{\curg\in\curG} U_\curg^\dag A U_\curg \,.
\end{equation}
Note that for every matrix $A$, $\Pi_\curG(A)$ commutes with all $U_\curg$ ($\curg\in\curG$). Thus, by Schur's lemma, since $\rho$ is irreducible%
\footnote{%
Schur's lemma guarantees this directly when $\loc=1$. But then it also applies for $\loc=2$ since 
then for any matrix $A=\sum_{i} B_i\otimes C_i$, we have $\Pi_\curG(A) = \frac{1}{|\curG|} \sum_{i} \sum_{a_1,a_2\in\F_q} U_{a_1}^\dag B_i U_{a_1} \otimes U_{a_2}^\dag C_i U_{a_2} \propto \sum_{i} \tr{B_i} \tr{C_i}= \tr{\sum_{i} B_i \otimes C_i} = \tr{A}$, and similarly for larger $\loc$.
}, we have $\Pi_\curG(A) = \frac{\tr(A)}{D}\id$ (where $D$ is the dimension of the Hilbert space). In particular then, if $\tr(A)=0$ then $\Pi_\curG(A)=0$.

\begin{protocol}[Bounded-strength balanced-cycle decoupling] \label{prot:inefficient-bounded-balanced}
Let $\curL$ be a balanced cycle on $\Gamma(\curG,\curS)$ of length $N = |\curG|\sum_\curs\mu_\curs = \lambda|\curG|$,  with group element representation $\curL_\curG = \left(\curg_0,\ldots,\curg_{N-1}\right)$ and generator representation $\curL_\curS = (\curs_1,\ldots,\curs_N)$.
For $j=1,\ldots,N$, set $U_c(0)=U_\cure=\id$ and 
\[
U_c\Big((j-1)\Delta+\delta\Big) = u_{\curs_j}\!(\delta) \ U_c\Big((j-1)\Delta\Big), \qquad \delta\in[0,\Delta].
\]
Note that because%
\footnote{up to phase, since $\rho$ is a projective representation; since we will only ever conjugate by $U_c$, the overall phase is irrelevant and we shall simply ignore it.} 
 $U_{\curs_j} U_{\curg_{j-1}} = U_{\curs_j + \curg_{j-1}} = U_{\curg_j}$, this implies 
$U_c(j\Delta)=U_{\curg_j}$ (for $j=0,\ldots,N$), i.e.
\begin{equation}\label{eqn:Uc_breakdown}
U_c\Big((j-1)\Delta+\delta\Big) = u_{\curs_j}\!(\delta) \ U_{\curg_{j-1}}, \qquad \delta\in[0,\Delta].
\end{equation}
The control cycle length is thus $T_c = N\Delta = |\curG|\lambda\Delta$.
\end{protocol}

\begin{theorem} \label{theorem:inefficient-bounded-balanced}
The above balanced-cycle protocol performs bounded-strength decoupling.
\end{theorem}

\begin{proof}
Since $\curL$ is a balanced cycle, $\leavenode{\curg}{\curs}$ occurs exactly $\mu_\curs$ times for every $\curg,\curs$ pair. Thus $u_\curs(\delta) U_\curg$ appears exactly $\mu_\curs$ times in the protocol for each $\curs, \curg$, and so we have, for any traceless $d^\loc \times d^\loc$ Hamiltonian $H$,
\begin{eqnarray*}
\bar H^{(0)} 
&=& \frac{1}{T_c} \int_{t=0}^{T_c} U_c(t)^\dag H U_c(t) dt \\
&=& \frac{1}{T_c}\sum_\curg U_\curg^\dag \left[\sum_\curs \mu_\curs\int_{\delta=0}^\Delta u_\curs(\delta)^\dag H u_\curs(\delta) d\delta \right] U_\curg \\
&=& \Pi_\curG\Big(F_\curS(H)\Big)
\end{eqnarray*}
where $\Pi_\curG$ is defined in Eq.~(\ref{eqn:Pi}) and $F_\curS$ is defined by
\begin{equation}\label{eqn:F}
F_\curS(H) = \sum_\curs \frac{\mu_\curs}{\lambda\Delta}\int_{\delta=0}^\Delta u_\curs(\delta)^\dag H u_\curs(\delta) d\delta\,.
\end{equation}
Recall that $\Pi_\curG$ suppresses traceless matrices. Assuming that $H$ is traceless, and observing that $F_\curS$ is trace-preserving, we have that
$\Pi_\curG\Big(F_\curS(H)\Big)=0$. We conclude that $\bar H^{(0)}=0$, i.e. the protocol succeeds at decoupling.
\end{proof}

\begin{remark} \label{remark:reducible}
For simplicity, we have assumed that $\rho$ is irreducible. Then this protocol works for any traceless time-independent $H$, even if $H$ is unknown. It is possible to define protocols in which $\rho$ is not irreducible, in which case $\Pi_\curG$ need not suppress all traceless matrices. However, in such a case, one must take special care to ensure that $\Pi_\curG$ still suppresses $F_\curS(H)$ for the Hamiltonians of interest. See \cite{eulerSim} for examples in a similar context, as well as Example~\ref{ex:diagonal} later in this paper.
\end{remark}

\begin{remark}\label{remark:sizes}
Although Protocol~\ref{prot:inefficient-bounded-balanced} performs bounded-strength decoupling, it would generally not be an efficient protocol were it applied to the entire system (i.e. if $\loc$ were the number of qudits of the entire system). Assuming that $\rho$ is irreducible, the representation 
$\rhot:\curG\rightarrow\unitaryset(d^\loc)$ necessitates that $|\curG|$, and therefore $T_c$, are exponential in $\loc$.
%
Indeed, suppose we have a representation from $\curG$ to $\unitaryset(D)$ such that for any $D\times D$ matrix $A$, $\Pi_\curG(A) = \frac{\tr(A)}{D}\id_D$ as we used in Theorem~\ref{theorem:inefficient-bounded-balanced}.
Consider sending the bipartite entangled state $|\psi\rangle = \frac{1}{\sqrt{D}}\sum_{j=1}^D |j\rangle \otimes |j\rangle$, or more precisely, $\Psi=|\psi\rangle\langle\psi|$, through the channel ${\cal I}\otimes\Pi_\curG$ (where ${\cal I}$ is the identity channel on a $D$-dimensional space) obtaining
$$\sum_{\curg\in\curG} \frac{1}{|\curG|} (\id_D\otimes U_\curg^\dag) \Psi (\id_D\otimes U_\curg)
=
({\cal I}\otimes\Pi_\curG)(\Psi) = \frac{1}{D^2}\id_D\otimes\id_D = \frac{1}{D^2}\id_{D^2} \,.$$
The matrix rank of the right-hand side is $D^2$. Using the fact that 
$\text{rank}(A+B) \leq \text{rank}(A) + \text{rank}(B)$ 
and that for each $\curg$, 
$\text{rank}(\frac{1}{|\curG|}(\id\otimes U_\curg^\dag) \Psi (\id\otimes U_\curg)) = \text{rank}(\Psi) = 1$, the rank of the left-hand side is at most $|\curG|$; thus,
$ |\curG| \geq  D^2. $
Therefore, for the representation $\rhot:\curG\rightarrow\unitaryset(d^\loc)$ to succeed in the proof of Theorem~\ref{theorem:inefficient-bounded-balanced}, we require that $|\curG| \geq d^{2\loc}$, which is exponential in $\loc$. 
Incidentally, by considering the case of $\loc=1$, we have justified why we could not have chosen $q$ less than $d^2$ in our irreducible representation $\rho :\F_q \rightarrow \unitaryset(d)$.
\end{remark}

Observe that the key to this protocol working is the fact that each $u_\curs(\delta)U_\curg$ shows up an equal number of times, independent of $\curg$, i.e. 
$\forall \curs\in\curS \ \exists \mu_s>0$ such that $\forall \curg\in\curG, \ \leavenode{\curg}{\curs}$ occurs $\mu_\curs$ times (independent of $\curg$). In an Eulerian cycle, $\mu_\curs=1$ for every $\curs$, which is certainly sufficient. All else being equal, given the choice between Eulerian and other balanced cycles, we would choose Eulerian cycles as they will minimize $N$ and therefore $T_c$. However, we will see that when considering the composite properties of a system (specifically that interactions are local), we will be able to exploit the notion of balanced cycles to come up with a much more efficient protocol.

\section{Balanced-cycle Orthogonal Arrays}

In Section~\ref{Sec:control_model} and Fig.~\ref{fig:array}, we indicated how we view our decoupling scheme as an array. For the protocol to be efficient, we shall ensure that this array corresponds to what we call a \emph{balanced-cycle orthogonal array} (BOA). A BOA is a special type of orthogonal array (OA), which we first define. We refer the reader to \cite{HSS99} for a thorough introduction to OAs, particularly their relationship to linear codes (of which we shall later make use).

For notational consistency, we point out that throughout the remainder of this paper we adopt the notation that $\curG$ and $\curS$ refer specifically to the group $\F_q^\loc$ and a generating set for $\F_q^\loc$, respectively. Elements of $\curG$ will be denoted using script $\curg$, elements of $\curS$ will be denoted using script $\curs$, and cycles on $\curG$ will be denoted $\curL$. When other groups (such as $\F_q$ or $\F_q^n$) are being considered, other notation (such as $g$, $m$, $S$, $s$ and $\calL$) will be used instead.

\begin{definition}[Orthogonal array]
An $OA_\lambda(N,n,q,\loc)$ \textit{orthogonal array} on the alphabet $\F_q$ is an $n\times N$ array where each of the $N$ columns is a vector from $\F_q^n$ such that every $\loc\times N$ subarray (obtained by only considering a selection of just $\loc$ of the $n$ rows) contains each possible $\loc$-tuple of elements of $\F_q$ (i.e. contains each $c\in\F_q^\loc$) precisely $\lambda$ times as a column. The number $\loc$ is called the \textit{strength} of the OA. 
\end{definition}

\begin{remark}
To relate these numbers to those appearing elsewhere in this paper, \begin{itemize}
\item $N$ will correspond to number of steps in the decoupling protocol (i.e. the length of our balanced cycle),
\item $n$ will correspond to the number of $d$-dimensional qudits describing the system, 
\item $q=d^2$ (e.g. for qubits, $d=2$ and $q=4$), 
\item $\loc$ is the locality of the Hamiltonian (e.g. for pairwise interactions, $\loc=2$), and
\item $\lambda = N/q^\loc$ will be the same $\lambda$ as in our discussion of balanced cycles, $\lambda=\sum_\curs \mu_\curs$.
\end{itemize}
\end{remark}

\begin{remark}\label{remark:OA_order}
Note that the order of the columns in the OA is irrelevant to whether the array is an OA. Moreover, if $A = [\vec a_i]$ is an $OA_\lambda(N,n,q,\loc)$ with columns $\vec a_1, \ldots, \vec a_N$ then the matrix $A'$, whose columns consist of precisely $r$ copies of each $\vec a_i$ (in any order), is an $OA_{r \lambda}(r N,n,q,\loc)$. Note, however, that while the order of the columns does not affect the OA property of the array, when defining balanced-cycle orthogonal arrays (which we do next), we will be highly concerned with the order of the columns in the array.
\end{remark}

\begin{figure}
	\begin{center}
		\includegraphics[width=2in]{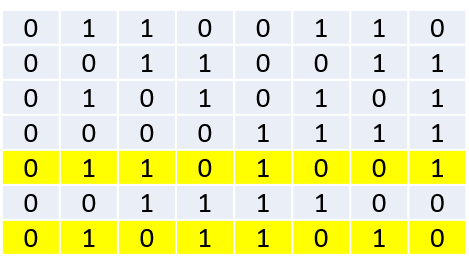}
		\caption[Example of a OA] 
		{\label{fig:OA}
		Example of a $OA_2(8,7,2,2)$, i.e. a $OA_\lambda(N,n,q,\loc)$ with $N=8$ columns and $n=7$ rows on the finite field $\F_q = \Z_2$ of order $q=2$. Any subarray defined by any $\loc=2$ rows contains each 2-tuple precisely $\lambda=2$ times. For example, rows 5 and 7 (highlighted) form a $2\times8$ subarray in which $\SmallMatrix{0\\0}, \SmallMatrix{0\\1}, \SmallMatrix{1\\0}, $ and $\SmallMatrix{1\\1}$ each occur precisely twice. Note that typically in this paper, $q=d^2$ (for example, for qubits $q=4$), but for simplicity, the example in this figure uses $q=2$.
		}
	\end{center}
\end{figure}
An example of an $OA_2(8,7,2,2)$ is shown in Fig.~\ref{fig:OA}. Orthogonal arrays have been used to construct bang-bang decoupling schemes (see \cite{stollsteimer,finite,equivalence} and \cite[Chapter 15]{QECbook}). In order to construct a bounded-strength scheme, we introduce the notion of a balanced-cycle orthogonal array, defined as follows.
\begin{definition}[Balanced-cycle orthogonal array]
A $BOA(N,n,q,\loc)$ \textit{balanced-cycle orthogonal array} on the alphabet $\F_q$ is an $n\times N$ array, $A$, where each of the $N$ columns is a vector from $\F_q^n$ such that every $\loc\times N$ subarray (obtained by only considering a selection of just $\loc$ of the $n$ rows) defines a balanced cycle on the Cayley graph of $\curG=\F_q^\loc$ with respect to some generating set for $\curG$ (which may depend on the subarray). Specifically, if the entries of $A$ are denoted $a_{ij}$ (with $1 \leq i \leq n$ and $0 \leq j \leq N-1$), then for every choice of $\loc$ distinct integers $i_1,\ldots,i_\loc\in\lbrace 1,\ldots, n\rbrace$, there is a generating set $\curS$ for $\curG$ (which may, in general, depend on $i_1,\ldots,i_\loc$) such that if $\curg_j = (a_{i_1 j},\ldots, a_{i_\loc j})^T$ denotes the $j$th column of $A$ restricted to rows $i_1,\ldots,i_\loc$, then $\curL_\curG = \left(\curg_0,\ldots,\curg_{N-1}\right)$ defines a balanced cycle on $\Gamma(\curG,\curS)$.
\end{definition}

\begin{figure}
	\begin{center}
		\includegraphics[width=5.5in]{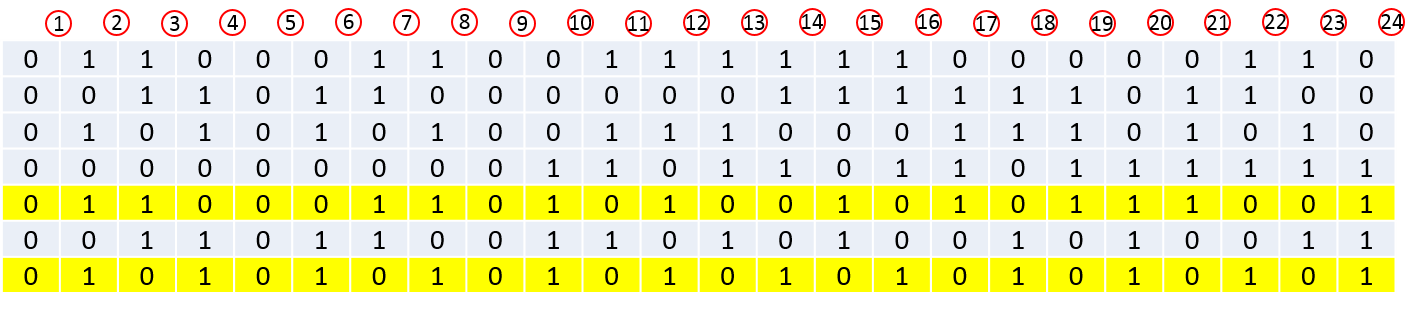}
		\caption[Example of a BOA]
		{\label{fig:BOA_highlight}
		Example of a $BOA(24,7,2,2)$, i.e. a $BOA(N,n,q,\loc)$ with $N=24$ columns and $n=7$ rows on the finite field $\F_q = \Z_2$ of order $q=2$. Any subarray defined by any $\loc=2$ rows defines a balanced cycle on the Cayley graph $\Gamma(\curG,\curS)$ of $\curG=\F_q^\loc = \Z_2^2$ with respect to some generating set $\curS$ (which may depend on the subarray). For example, rows 5 and 7 (highlighted) form a $2\times24$ subarray that defines the balanced cycle shown in Fig.~\ref{fig:balanced_cycle_hourglass}. The circled integers (red) correspond to the steps taken by the balanced cycle as shown in that figure. Note that typically in this paper, $q=d^2$ (for example, for qubits $q=4$), but for simplicity, the BOA example shown here uses $q=2$. The method by which this BOA was constructed is detailed in Example~\ref{ex:pedagogical} of Sec.~\ref{sec:examples}.
		}
	\end{center}
\end{figure}
An example of a BOA is shown in Fig.~\ref{fig:BOA_highlight}. We defer the proof that BOAs exist to Sec.~\ref{sec:BOA_exist}. The remainder of the current section defines a decoupling protocol based on BOAs and proves that it works to decouple $\loc$-local Hamiltonians in $n$ qudit systems ($\loc\leq n$). 
Working with $\loc$ qudits (rather than $n$ qudits), along with the promise that $H$ is $\loc$-local, will enable us to give an efficient protocol.

\begin{protocol}[Efficient, bounded-strength balanced-cycle decoupling based on BOAs] \label{prot:BOA}
Let $A = [\vec a_j]_{j=0,\ldots,N-1}$ be a $BOA(N,n,q,\loc)$ whose columns are denoted by the vectors $\vec a_j = (a_{1 j},\ldots, a_{n j})^T$, where $a_{ij}\in\F_q$ is the $(i,j)$ entry of $A$. For $j=1,\ldots,N$, let $\vec b_j = \vec a_j - \vec a_{j-1}$ be the transitions between the columns, treating ${\vec a_N} = {\vec a_0} = 0$.

For $j=1,\ldots,N$, set $U_c(0)=\id$ and 
\[
U_c\Big((j-1)\Delta+\delta\Big) = u_{\vec b_j}(\delta) U_c\Big((j-1)\Delta\Big), \qquad \delta\in[0,\Delta];
\] note that this implies that $U_c(j\Delta)=U_{\vec a_j}$ (for $j=0,\ldots,N$). 
The control cycle length is thus $T_c = N\Delta$.
\end{protocol}

\begin{theorem} \label{theorem:BOA}
The above protocol performs bounded-strength decoupling.
\end{theorem}

\begin{proof}

$H$ is an $\loc$-local Hamiltonian, $H = \sum_\Hamindex H_\Hamindex$ with each $H_\Hamindex$ acting non-trivially on at most $\loc$ qudits.
Consider a term $H_\Hamindex$, which acts non-trivially only on qudits denoted $i_1,\ldots,i_\loc$ and write $H_\Hamindex = h_\Hamindex \otimes \id_{n-\loc}$, where $h_\Hamindex$ is understood to be a $d^\loc\times d^\loc$ matrix acting only on these $\loc$ qudits and $\id_{n-\loc}$ is the identity matrix on the other $n-\loc$ qudits. 
By definition of a BOA, the $\loc\times N$ subarray of $A$ restricted to rows $i_1,\ldots,i_\loc$ defines a balanced cycle $\curL$ on $\Gamma(\curG,\curS)$ where $\curS$ is some generating set of $\curG=\F_q^\loc$. 
The idea of the proof is to observe that the protocol involving the columns $\vec a_j$ for decoupling $H_k$ is equivalent to a protocol involving the subarray's columns for decoupling $h_k$; since the subarray defines a balanced cycle, we can then invoke Protocol~\ref{prot:inefficient-bounded-balanced} to successfully decouple $h_k$ and therefore $H_k$.

Let $\curg_j = (a_{i_1 j},\ldots, a_{i_\loc j})^T$ denote the $j$th column of $A$ restricted to rows $i_1,\ldots,i_\loc$ and let $\curs_j = \curg_j - \curg_{j-1} = (b_{i_1 j},\ldots, b_{i_\loc j})^T$, where $b_{ij}$ is the $i$th entry of $\vec b_j$. Then the cycle $\curL$ is represented as $\curL_\curG = \left(\curg_0,\ldots,\curg_{N-1}\right)$ and $\curL_\curS = (\curs_1,\ldots,\curs_{N})$.

As in the proof of Theorem~\ref{theorem:inefficient-bounded-balanced}, we are interested in $U_c(t)^\dag H_\Hamindex U_c(t)$. The control unitary at time $t=(j-1)\Delta+\delta$ is
\[
U_c\Big((j-1)\Delta+\delta\Big) 
= u_{\vec b_j}\!(\delta) \, U_c\Big((j-1)\Delta\Big)
= u_{\vec b_j}\!(\delta) \, U_{\vec a_{j-1}}
= \left( u_{b_{1j}}(\delta) \otimes\cdots\otimes u_{b_{nj}}(\delta) \right)
	\left( U_{a_{1 (j-1)}} \otimes\cdots\otimes U_{a_{n (j-1)}} \right) \,.
\]
Thus, when conjugating $H_\Hamindex = h_\Hamindex \otimes \id_{n-\loc}$ by $U_c\big((j-1)\Delta+\delta\big)$, all of the unitaries not acting on the $\loc$-qudit subspace of $h_\Hamindex$ will commute through $H_\Hamindex$ and cancel, leaving only those corresponding to the $\loc$-qudit subspace, i.e. those corresponding to the labels $\curs_j$ and $\curg_j$. 
Explicitly,
\begin{eqnarray*}
U_c\Big((j-1)\Delta+\delta\Big)^\dag {\ H_\Hamindex \ } U_c\Big((j-1)\Delta+\delta\Big)
&=&
	\bigg[\left( U_{a_{i_1 (j-1)}}^\dag \otimes\cdots\otimes U_{a_{i_\loc (j-1)}}^\dag \right)
		\left( u_{b_{i_1 j}}(\delta)^\dag \otimes\cdots\otimes u_{b_{i_\loc j}}(\delta)^\dag \right)
			{\ h_\Hamindex \ }\bigg.  \\
&&\qquad
		\bigg.\left( u_{b_{i_1 j}}(\delta) \otimes\cdots\otimes u_{b_{i_\loc j}}(\delta) \right)
	\left( U_{a_{i_1 (j-1)}} \otimes\cdots\otimes U_{a_{i_\loc  (j-1)}} \right)\bigg]
	 \otimes \id_{n-\loc}\\
&=&
U_{\curg_{j-1}}^\dag u_{\curs_j}(\delta)^\dag {\ h_\Hamindex \ } u_{\curs_j}(\delta) U_{\curg_{j-1}} \otimes \id_{n-\loc} \,.
\end{eqnarray*}
 
Thus, the protocol of applying $U_c$ to $H_\Hamindex$ is effectively the same as applying a protocol $u_{\curs_j}(\delta) U_{\curg_{j-1}}$ to $h_\Hamindex$, following the balanced cycle $\curL$. Since this is precisely the scheme defined in Protocol~\ref{prot:inefficient-bounded-balanced} applied to $h_\Hamindex$ (see Eq.~(\ref{eqn:Uc_breakdown})), we conclude from Theorem~\ref{theorem:inefficient-bounded-balanced} that it decouples $h_\Hamindex$.
 Consequently, $\bar H_\Hamindex^{(0)} = \bar h_\Hamindex^{(0)} \otimes \id_{n-\loc}  = 0$. This occurs for every term $H_\Hamindex$ in $H=\sum_\Hamindex H_\Hamindex$, whence $H$ itself is decoupled: $\bar H^{(0)} = \sum_\Hamindex \bar H_\Hamindex^{(0)} = 0$.

\end{proof}

\begin{remark} \label{remark:increase_n}
Once we have a BOA scheme that can decouple a system of $n$ qudits, the same scheme can be used (with the same BOA and therefore same length $N$) for a system of $n' < n$ qudits. This can be accomplished by simply ignoring $n-n'$ of the qudits, i.e. by having $U_c$ act as $\id$ on these $n-n'$ extra qudits (rather than as dictated by the original protocol). The proof of Theorem~\ref{theorem:BOA} remains unaffected because $H_k$ acts trivially on these extra qudits, i.e. they are not acted upon by $h_k$.
\end{remark}

Theorem~\ref{theorem:BOA} showed that decoupling protocols based on BOAs work, with control cycle length proportional to the BOA parameter $N$. We next show that BOAs can indeed be constructed and, moreover, that the construction gives rise to an \textit{efficient} decoupling protocol, in the sense that $N$ does not increase exponentially with $n$.

\section{Construction of balanced-cycle orthogonal arrays} \label{sec:BOA_exist}

The existence of balanced-cycle orthogonal arrays follows naturally from constructions of orthogonal arrays generated using classical linear codes, which we shall define shortly. We first give a brief outline of our BOA construction. This construction is via the generator matrix $G$ of a linear code, which is a linear mapping from $\F_q^k$ to $\F_q^n$ for some $k \leq n$. 
If we enumerate all elements of $\F_q^k$ in an arbitrary order and consider their image under $G$, this will form an OA of strength $\loc$ (for an appropriately chosen $k$). To obtain a BOA, we do this enumeration according to the prescription of an Eulerian cycle on $\F_q^k$. In doing so, we can guarantee that we always obtain a balanced cycle when we consider any submatrix of $\loc$ rows, ultimately ensuring that any $\loc$-local Hamiltonian term on those corresponding qudits will be decoupled. We now prove this, starting with a definition of a classical linear code.

\begin{definition}[Classical linear code]
A \textit{classical linear} $[n, k]_q$ \textit{code}, $C$,  is a $k$-dimensional
subspace of the vector space $\F_q^n$. 
For any vector $x = (x_1, \ldots, x_n)^T \in \F_q^n$, define 
${\rm wt}(x) = |\{ i\in \{1, \ldots, n\} : x_i \not=0\}|$. The
 \textit{distance} of a linear code $C$ is defined to be $\min\{{\rm wt}(c) : c \in C, c\neq {o}\}$, where ${o}$
denotes the zero vector.
An $[n,k]_q$ linear code can be described by a \textit{generator matrix} $G$ of size $n\times k$ with entries from $\F_q$. 
$G$ maps the vectors $m\in\F_q^k$ onto the elements (\textit{code words}) of $C$ so that $C = G[\F_q^k] = \lbrace Gm \in \F_q^n : m\in \F_q^k  \rbrace$.
\end{definition}

The dual code $C^\perp$ of $C$ is defined by $C^\perp = \{ y \in
\F_q^n : x \cdot y = 0\  \forall x \in C\}$ with the dot product $x\cdot y = \sum_{i=1}^n x_i y_i$. The dual code is also a classical linear code, namely an $[n,n-k]_q$ code with some distance $\dist^\perp$ that we will refer to as the \textit{dual distance}.
Orthogonal arrays can be constructed from linear codes, as the following theorem \cite[Theorem 4.6]{HSS99} establishes. 
\begin{theorem}[OAs from linear codes]\label{theorem:OA_codes}${}$
Let $C$ be a linear $[n, k]_q$ code with dual distance
$\dist^\perp$. The $n\times q^k$ matrix $[Gm]_{m\in\F_q^k}$, whose columns are the $q^k$ vectors $Gm\in\F_q^n$ $(\forall m\in\F_q^k)$, is an $OA(q^k,n,q,\loc)$ with strength $\loc=\dist^\perp-1$.
\end{theorem}

Let $C$ be an $[n,k]_q$ with dual distance $\dist^\perp = \loc+1$ and generating matrix $G$.
Let $\calL$ be an Eulerian cycle on the Cayley graph $\Gamma(\F_q^k,\Sqk)$,  where $\Sqk$ is a generating set for $\F_q^k$; thus, we can write $\calL_{\F_q^k} = (m_0,\ldots,m_{N-1})$ with transitions $\calL_{\Sqk} = (s_1,\ldots,s_{N})$ and $N=q^k\,|\Sqk|$. 
Because $d$ is a prime power, say $d=p^e$ for some prime $p$, the minimal generating set is of size $|\Sqk| = 2ke$.
We are interested in the image of the cycle in the codespace; thus, consider the Eulerian cycle, denoted $G\calL$, on $\Gamma(G[\F_q^k],G[\Sqk])$, where $G[\F_q^k]=C \subset \F_q^n$ is the image of $\F_q^k$ under $G$, i.e. is the codespace.
In other words, $G\calL_{\F_q^k} = (Gm_0,\ldots,Gm_{N-1})$ and $G\calL_{\Sqk} = (Gs_1,\ldots,Gs_{N})$.

To avoid possible confusion, we emphasize here that although we will use $G\calL$ to construct a BOA, neither $\calL$ nor $G\calL$ will serve as the balanced cycle to which Theorem~\ref{theorem:inefficient-bounded-balanced} applies (which is why we have used the notation $\calL$ rather than $\curL$). Rather, for our efficient decoupling scheme, we construct an array $A_{G\calL}$, dictated by $G\calL$, and prove that the result is a BOA by showing that if we consider any subarray of $\loc$ rows, it gives rise to some balanced cycle $\curL$ on $\curG=\F_q^\loc$. The notation and relationships of the various groups and cycles used in this paper is sketched in Fig.~\ref{fig:groups}.

\begin{figure}
	\begin{center}
		\includegraphics[width=6in]{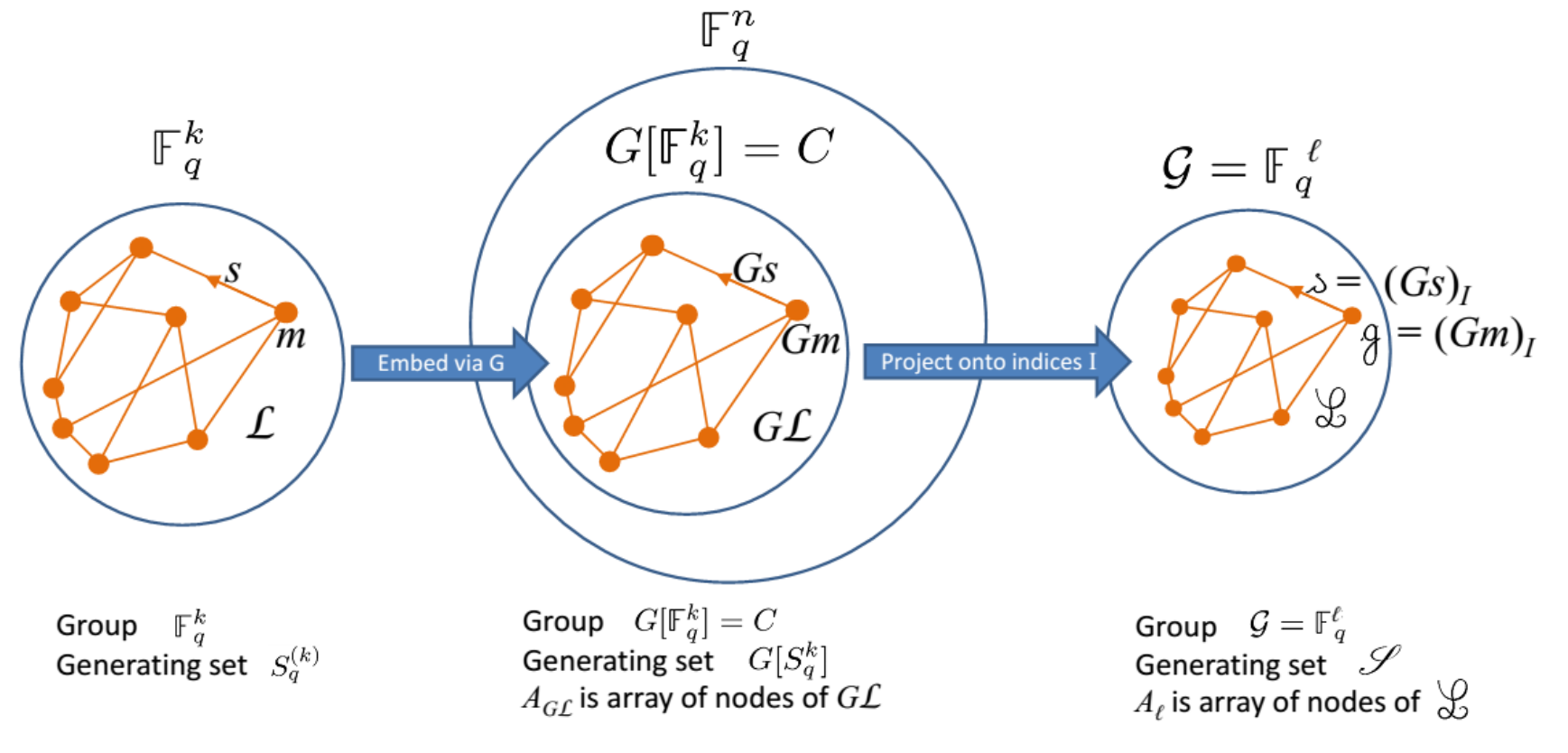}
		\caption
		{\label{fig:groups}
		Names and relationships between various groups and cycles in the BOA construction. The graph is a schematic of a Cayley graph. As explained in the text, the BOA, $A_{G\calL}$, is the array form of $G\calL$, which is the result of mapping an Eulerian cycle $\calL$ under the linear code generating matrix $G$. As a BOA, $A_{G\calL}$ has the property that if one considers a subarray of $\loc$ rows, the result describes a balanced cycle. Specifically, let $I \subset \lbrace 1,2,\ldots,n \rbrace$ be a subset of $\loc$ indices. $(Gm)_I$ denotes the $\loc$-tuple of elements of $Gm$ (itself an $n$-tuple) corresponding to the indices $I$. The cycle $\curL$, composed of nodes $(Gm)_I$ (in the same order in which $\calL$ was composed of $m$), is shown to be a balanced cycle.
		}
	\end{center}
\end{figure}

We turn $G\calL$ into an array $A_{G\calL}$ in the obvious way as follows.
Each element $Gm_j$ of $G\calL_{\F_q^k}$ is a column vector in $\F_q^n$. Therefore we may associate to $G\calL$ the $n\times N$ matrix $A_{G\calL} = [Gm]_{m\in\calL_{\F_q^k}}$ with elements $a_{ij} = (Gm_j)_i$, so that the $j$th column of $A_{G\calL}$ is the vector $Gm_j$, and the columns are arranged in the order of the Eulerian cycle $G\calL$.
Note that since we assumed that Eulerian cycles always start with the (additive) identity element, i.e. the zero vector $o\in\F_q^k$, and since $G$ maps the zero vector to the zero vector ($Go = o\in\F_q^n$), the first column of $A_{G\calL}$ is the zero vector of $\F_q^n$.
\begin{lemma} 
$A_{G\calL}$ is an $OA_{N/q^\loc}(N,n,q,\loc)$ with $N=q^k\,|\Sqk|$.
\end{lemma}
\begin{proof}
By Theorem~\ref{theorem:OA_codes}, an array whose $q^k$ columns are the vectors of the codespace is an OA. The columns of $A_{G\calL}$ are precisely $|\Sqk|$ copies of each vector in the codespace, and therefore (using Remark~\ref{remark:OA_order}), $A_{G\calL}$ is an $OA$.
\end{proof}

Let $s\in \Sqk$. $Gs\in\F_q^n$, so $Gs=\Big((Gs)_1,\ldots,(Gs)_n \Big)^T$, where we use the notation $(Gs)_i\in\F_q$ to denote the $i$th component of the column vector $Gs$. 
Fix $\loc$ distinct numbers $i_1,\ldots,i_\loc\in \lbrace 1,\ldots,n \rbrace$ and write $I = \lbrace i_1,\ldots,i_\loc \rbrace$. Let $(Gs)_I$ denote the $\loc$-tuple $\Big((Gs)_{i_1},\ldots,(Gs)_{i_\loc}\Big)^T$.
Let $\curS = \left\lbrace (Gs)_I : s\in \Sqk \right\rbrace$. 
\begin{lemma} \label{lemma:generating}
$\curS$ is a generating set for $\curG=\F_q^\loc$.
\end{lemma}
\begin{proof}
Let $\curg\in\curG$. By definition, since $A_{G\calL} = [Gm]_{m\in\calL_{\F_q^k}}$ is an OA of strength $\loc$, the $\loc\times N$ subarray obtained by only considering rows $i_1,\ldots,i_\loc$ contains each possible $\loc$-tuple of elements of $\F_q$, and therefore contains $\curg$. Thus, $\exists \ Gm$ such that $(Gm)_I = \curg$.
Since $\Sqk$ is a generating set for $\F_q^k$, $\exists \ u_1,\ldots,u_r\in \Sqk$ such that $m = u_1+\cdots +u_r$, and therefore $Gm = Gu_1+\cdots +Gu_r$. But then $(Gu_j)_I\in\curS$ for every $j=1,\ldots,r$ and $\curg = (Gm)_I = (Gu_1)_I+\cdots +(Gu_r)_I$, whence $\curS$ generates $\curG$.
\end{proof}

Recall $A_{G\calL}=[Gm]_{m\in\calL_{\F_q^k}}$ and consider the $\loc\times N$ submatrix $A_\loc = [\curg_{j}]$ of $A_{G\calL}$, whose $j$th column is $\curg_j = (Gm_j)_I \in \curG$. 
Define the ordered list $\curL_\curG =  \left(\curg_0,\ldots,\curg_{N-1}\right)$.
Although $\curL$ depends on $I$, we suppress mention of this for notational simplicity.
\begin{lemma} \label{lemma:cycle}
$\curL$ is a balanced cycle on $\Gamma(\curG, \curS)$.
\end{lemma}
\begin{proof}
 Each $\curg\in\curG$ is present in $\curL_\curG$ an equal number of times because $A_{G\calL}$ is an OA of strength $\loc$.  
 The transitions in this cycle are $\curs_j = \curg_j - \curg_{j-1} = (Gm_j)_I - (Gm_{j-1})_I = (Gs_j)_I \in \curS$, so the transition representation $\curL_\curS =  \left(\curs_1,\ldots,\curs_{N}\right)$ consists of generators from $\curS$; $\curL$ is therefore a cycle on the Cayley graph $\Gamma(\curG, \curS)$. Moreover, because $\calL$ is an Eulerian cycle and $A_{G\calL}$ is an OA, $\curL$ is a balanced cycle (although not an Eulerian cycle): informally, each $\leavenode{Gm}{Gs}$ occurs in $G\calL$ an equal (non-zero) number of times (namely once, independent of $Gm$) for each $Gs$, so each $\leavenode{\curg}{\curs}$ occurs in $\curL$ an equal (non-zero) number of times (independent of $\curg=(Gm)_I$, since $A_{G\calL}$ is an OA) for each $\curs = (Gs)_I$.

Explicitly, consider any $\curg\in\curG$ and $\curs\in\curS$. 
Let $M_\curg = \lbrace m\in\F_q^k : (Gm)_I = \curg \rbrace$. $\calL_{\F_q^k}$ is an Eulerian cycle so each element in $\F_q^k$ shows up precisely $|\Sqk|$ times in $\calL_{\F_q^k}$. In particular, therefore, each $m\in M_\curg$ appears precisely $|\Sqk|$ times in $\calL_{\F_q^k}$, and consequently, $\curg$ shows up in $\curL$ precisely $|M_\curg| |\Sqk|$ times. But $A_{G\calL}$ is an OA of strength $\loc$, so $|M_\curg| |\Sqk|$ must then be independent of $\curg$, and therefore $|M_\curg|$ is also independent of $\curg$.
Since $\curs\in\curS$, let $S_\curs = \lbrace s\in \Sqk : (Gs)_I = \curs \rbrace$. This set is non-empty by definition of $\curS$. In general, $|S_\curs|$ may depend on $\curs$.
Now, $\forall m\in M_\curg$ and $\forall s\in S_\curs$, the Eulerian property of $\calL$ guarantees that $\leavenode{Gm}{Gs}$ occurs precisely once in $G\calL$. Therefore, $\leavenode{\curg}{\curs}$ occurs in $\curL$ precisely $|S_\curs| |M_\curg| \geq 1$ times, which is independent of $\curg$. Thus $\curL$ is a balanced cycle.
\end{proof}

Together, the above lemmas prove the existence of BOAs and how to construct them from classical linear codes.
\begin{theorem} \label{theorem:BOA_exist}
Let $C$, $\calL$, and $A_{G\calL}$ be as above, i.e.
$C$ is an $[n,k]_q$ code with dual distance $\dist^\perp = \loc+1$ and generating matrix $G$, $\calL$ is an Eulerian cycle on the Cayley graph $\Gamma(\F_q^k,\Sqk)$, written $\calL_{\F_q^k} = (m_0,\ldots,m_{N-1})$ and $\calL_{\Sqk} = (s_1,\ldots,s_{N})$, and
$A_{G\calL} = [Gm]_{m\in\calL_{\F_q^k}}$ is an OA whose columns are the vectors $Gm_j$.
Then $A_{G\calL}$ is a $BOA(N,n,q,\loc)$ with $N=q^k|\Sqk|$.
\end{theorem}

\begin{proof}
For every choice of $\loc$ distinct integers $I = \lbrace i_1,\ldots,i_\loc \rbrace \subset \lbrace 1,\ldots, n\rbrace$, the set $\curS = \left\lbrace (Gs)_I : s\in \Sqk \right\rbrace$ is a generating set for $\curG$ (by Lemma~\ref{lemma:generating}) such that if $\curg_j$ denotes the $j$th column of $A_{G\calL}$ restricted to rows $i_1,\ldots,i_\loc$, then (by Lemma~\ref{lemma:cycle}) $\curL_\curG = \left(\curg_0,\ldots,\curg_{N-1}\right)$ defines a balanced cycle on $\Gamma(\curG,\curS)$.
\end{proof}

For $n$ interacting qudits of dimension $d=p^e$ (for some prime $p$ and positive integer $e$) that obey an $\loc$-local Hamiltonian,
this construction therefore allows 
\begin{equation}\label{eqn:N}
N=q^k |\Sqk| = q^{k} 2ke
\end{equation}
where $k$ is the dimension of the code used and $q=d^2$.
Observe that the BOA decoupling protocol (Protocol~\ref{prot:BOA}) for this BOA construction has a control cycle length of $T_c = N\Delta = d^{2k} 2ke \Delta$ where $\Delta$ is some fixed length of time. 
For example, in the qubit ($d=2$) case discussed above, $|\Sqk[4]| = 2k$, whence $T_c = N\Delta$ with $N= (2k) 4^k$. 
To maximize efficiency for a given $n$ and $\loc$, one should select a code that minimizes $k$ (equivalently, select a dual code that maximizes $k^\perp = n-k$). 

There exist many good families of classical linear codes.
For instance, for 2-local interactions, we can (as was done in \cite{Roetteler:2008} for OAs) rely on  $[n,k]_q$ Hamming codes with dual distance 3 such that $k = \log_q\big((q-1)n+1\big)$; our scheme then has $N$ scaling like $n\log(n)$.
This protocol is therefore much more efficient than a naive protocol of applying balanced-cycle decoupling (including Eulerian decoupling) without exploiting the $\loc$-local structure of the Hamiltonian, which would have a control cycle length that scales exponentially with $n$. 
It is also more efficient than the method of \cite{eulerOA}, which required $N = d^{4k}$, i.e. whose scaling for this case 
($\loc=2$, using Hamming codes) is quadratic in $n$.
Next, we address codes for BOA construction with values of $\loc$ greater than 2.

\section{BOA decoupling schemes from BCH codes}

In this section we show how to construct schemes that achieve decoupling for $\loc$-local Hamiltonians on ${\cal H} \cong (\C^d)^{\otimes n}$ for arbitrary $\loc$, $n$, and prime power $d$. Besides the machinery of balanced-cycle orthogonal arrays (BOAs) that was introduced in the previous sections, our construction relies on BCH codes as a particular vehicle to construct good BOAs. The choice of BCH codes results from the fact that they are among the best known codes for the particular situation  where the distance is a fixed, small number and the goal is to maximize the overall code dimension. Using the dual of a BCH code when constructing the corresponding orthogonal arrays, we obtain schemes with a designed OA strength (i.e. locality $\loc$) while having a small $N$ in the corresponding decoupling protocol. We begin by briefly recalling some basics about BCH codes; for more details on finite fields and BCH codes see, for example, the textbooks \cite{MS:77,LC:2004,Roth:2006}.

\begin{definition}[BCH code]\label{def:bch}
Let $\alpha$ be a primitive $n$-th root of unity in the finite field $\F_{q^m}$, where $q$ is a prime power, $n\geq 2$, and $m \geq 1$. A \emph{BCH code} over $\F_q$ of length $n$ and designed distance $D$, where $2 \leq D \leq n$, is a cyclic polynomial code defined by the zeros 
\[ 
\alpha^b, \alpha^{b+1}, \ldots, \alpha^{b+D-2},
\]
where $b\geq 1$ is a positive integer.
\end{definition}

The generator polynomial $g(x)$ of the cyclic code introduced in Definition \ref{def:bch}
is given by $g(x) = \text{lcm}(M_b(x), M_{b+1}(x), \ldots, M_{b+D-2}(x))$, where $M_i(x)$ denotes the minimal polynomial of $\alpha^i$ over $\F_q$. Note that even though the zeros of the code lie in an extension field $\F_{q^m}$ over $\F_q$, the BCH code itself is a cyclic code over the ground field $\F_q$. Furthermore, it is known that a BCH code defined this way has a distance $\dist$ that is at least $D$, which is why $D$ is sometimes called the ``designed distance.'' Note that the actual distance $\dist$ of the code might exceed $D$. The possible lengths of BCH codes are quite restricted, as any admissible length $n$ must be a divisor of the order of the multiplicative group of $\F_{q^m}$, i.e. must be a divisor of $q^m-1$. In the following we restrict ourselves to the case where $n=q^m-1$, which is called the case of \textit{primitive BCH codes}. Furthermore we only consider the case where $b=1$, which is called the case of \textit{narrow-sense BCH codes}. We denote these codes by ${\rm BCH}(\F_{q^m}/\F_q, D)$, and we note that they always exist.

For any linear error-correcting code $C=[n,k,\dist]_q$ of length $n$, dimension $k$, and distance $\dist$, an \textit{extension} $C'=[n+1,k,\dist' \geq \dist]_q$ can be defined by adding another coordinate and an overall parity check. At the level of parity check matrices, this corresponds to appending the 
parity check matrix $M$ of $C$ with an all-zeros column ${\mathbf 0}$ and an all-ones row ${\mathbf 1}^T$ so that $C'$ has the new parity check matrix 
$\left[ \begin{array}{cc} {\mathbf 1}^T & 1 \\ M & {\mathbf 0} \end{array} \right]$. For binary codes, the distance of the extension is easy to characterize: if $\dist\equiv 0 \mod  2$ then $\dist'=\dist$ and if $\dist\equiv 1 \mod 2$ then $\dist'=\dist+1$. In general over larger alphabets, however, it is possible that the distance increases  even when $\dist$ is even. When applying an extension to the BCH codes introduced above, we use the notation ${\rm BCH^{ext}}(\F_{q^m}/\F_q, D)$. We make use of the following theorem about such codes.

\begin{theorem}\label{thm:bchbound}
Let $\F_q$ be a finite field and let ${\rm BCH^{ext}}(\F_{q^m}/\F_q,D)=[n,k,\dist]_q$ be the extension of the primitive narrow-sense BCH code with designed distance $D$ constructed in Definition \ref{def:bch}, so $n=q^m$ and $\dist\geq D$. Assume that $D \leq q^{\lceil m/2 \rceil} +2$. Then the dimension $k$ of the code satisfies $k \geq n - m \big\lceil \frac{q-1}{q} (D-2)\big\rceil - 1 \geq n - m(D-2)-1$.
\end{theorem}

See \cite{Sudan:2001} for a proof of Theorem \ref{thm:bchbound} that leverages the fact that the extended primitive narrow-sense BCH codes are subfield subcodes of the Reed-Solomon codes. See also \cite[Problem 8.12]{Roth:2006} and \cite{YD:2004}. By combining Theorem \ref{thm:bchbound} with the construction of Theorem \ref{theorem:BOA_exist} we now obtain the following result regarding bounded-strength decoupling for $\loc$-local Hamiltonians. 

\begin{theorem}\label{thm:decouplinglocal}
For any $\loc\geq 2$, $n \geq (\loc-1)^2$, and $q=d^2$ with $d\geq 2$ a prime power, there exists a $BOA(N,n,q,\loc)$ whose length $N$ scales as $N = O(n^{\loc-1} \log n)$. That is, there exists a bounded-strength BOA decoupling scheme to switch off $\loc$-local Hamiltonians on $n$ interacting $d$-dimensional qudits that uses $N = O(n^{\loc-1} \log n)$ time slices.
\end{theorem}
\begin{proof}
First, note that if $n$ is not of the special form $n=q^m$ where $m\geq 1$, then we can always embed the $n$ qudits into a larger system of $q^m$ qudits with $m=\lceil \log_q(n) \rceil$, construct a scheme for the larger system, and ignore the additional qudits (as per Remark~\ref{remark:increase_n}). This increases $n$ by a factor of at most $q$ and therefore doesn't affect the statement of the theorem, i.e. we can without loss of generality assume that $n=q^m$ where $m\geq 1$. 

Now, we consider the code $C$ that is the \textit{dual} of a $k^\perp$-dimensional ${\rm BCH^{ext}}(\F_{q^m}/\F_q,D)$ code with designed distance $D=\loc+1$. Thus $C$ has length $n$, dual distance $\dist^\perp \geq D = \loc+1$, and, according to Theorem \ref{thm:bchbound}, dimension $k = n-k^\perp \leq m (D-2) +1 = m(\loc - 1) +1$. By Theorem~\ref{theorem:OA_codes}, this means that we can construct an $n \times N_{OA}$ orthogonal array of strength $\dist^\perp-1\geq\loc$ from this code, where $N_{OA}=q^{k} \leq q^{m (\loc-1) + 1} = q n^{\loc-1}$. According to Theorem \ref{theorem:BOA_exist}, the corresponding BOA has an overhead that scales at most logarithmically in $n$ since from Eq.~(\ref{eqn:N}) we obtain the following bound on the length of the bounded-strength decoupling scheme corresponding to the BOA: $N_{BOA}=q^{k} |\Sqk| = q^{k} 2ke \leq [qn^{\loc-1}] [2(m(\loc-1)+1)e] = 2qen^{\loc-1} [(\loc-1)\log_q n + 1] = O(n^{\loc-1} \log n)$. This establishes the claimed bound.  
\end{proof}

In physical systems, the locality $\loc$ is generally a small fixed number, so the requirement of $n \geq (\loc-1)^2$ is inconsequential asymptotically, while for small $n$, one can (by Remark~\ref{remark:increase_n}) always artificially increase $n$ to satisfy it.
Our main focus in Theorem \ref{thm:decouplinglocal} is on the asymptotic cost for fixed locality $\loc$ as the number $n$ of qudits grows. It should be noted that, depending on the particular choice of $q$, $\loc$, and $n$, further improvements over the bound in Theorem \ref{thm:decouplinglocal} are possible; see e.g., \cite{YD:2004,EB:2001}. This in turn leads to further improvements in the length of the decoupling schemes constructed via Theorem \ref{theorem:BOA_exist}. 
For instance, for $2$-local qubit Hamiltonians we saw at the end of Sec.~\ref{sec:BOA_exist} that Hamming codes can be used to construct BOA decoupling schemes of length $N = 2[3n + 1]\log_4[3n + 1]$, giving a slight improvement over schemes constructed from primitive BCH codes which lead to a scaling of $N \leq 8n[\log_4(n) + 1]$.

\section{Tables of best known BOA schemes for small systems}

In the following, we present a summary of the best known BOA schemes for qubit ($d=2$, $q=4$) and qutrit ($d=3$, $q=9$) systems for a variety of small localities $\loc$ and system sizes $n$. All schemes are obtained by our main construction in Theorem \ref{theorem:BOA_exist}, where the underlying classical linear codes are either taken from the literature or from the Magma \cite{Magma} database of best known linear codes which can be accessed using the Magma command \texttt{BestDimensionLinearCode(<field>, <length>, <distance>)}. 

Recall from Remark~\ref{remark:increase_n} that if we have a BOA decoupling scheme for $n$ qudits, it can also be used for smaller systems of $n'<n$ qudits. Therefore, the best known BOA for $n$ qudits is also the best known BOA for all $n'<n$ qudits unless a better BOA scheme for $n'$ is known. Table \ref{tab:gf4} summarizes the best known schemes for systems of $n$ qubits ($d=2$), for small values of $n$, that can be obtained from good linear codes. Similarly, Table \ref{tab:gf9} summarizes the best known schemes for systems of $n$ qutrits ($d=3$), for small values of $n$.

\begin{table}
\begin{tabular}{c@{\quad}@{\quad}c@{\qquad}c@{\qquad}c@{\qquad}c@{\qquad}c@{\qquad}c@{\qquad}c}
\hline
$\ell \; \backslash \; N$  &  
  $64$  & $384$& $2\,048$ & 10\,240        & 49\,152         	   &  229\,376      		   &  1\,048\,576		  			\\
\hline
2 & 2--$5^a$& 6--$21^a$& 22--$85^a$& 86--$341^a$	& 342--$1\,365^a$ & 1\,366--$5\,461^a$ & 5\,462--$21\,845^a$ \\
3 & - 	 & 3--$6^b$ & 7--$17^c$ & 18--$41^c$ 	& 42--$126^c$				   & 127--$288^c$				   & 289--$756^c$  				 \\
4 & - 	 & - 	 & 4--5 	  & 6--$11^d$ 		& 12--$21^e$ 				   & 22--43  				   & 44--85 					 \\
5 & - 	 & - 	 & - 	  & 5--6 			& 7--$12^f$  				   & 13--20  				   & 21--27 				\\
6 & - 	 & - 	 & - 	  & - 			& 6--7 	  				   & 8--9	 					& 10--17 			 \\
7 & - 	 & - 	 & - 	  & - 			& - 	  					& 7--8   	 					& 9--10      			 \\
8 & - 	 & - 	 & - 	  & - 			& - 	  					& - 		 				& 8--9 				\\
\hline
\end{tabular}
\caption{\label{tab:gf4} Table of the best known balanced-cycle orthogonal arrays (BOAs) for qubit ($d=2$) systems, indicating the number of qubits that can be decoupled by a BOA scheme for the given locality and length. Shown are the locality $\ell$ of the underlying Hamiltonian from $2$ up to $8$ and length $N=4^k 2k$ of the BOA cycles from $64$ up to $1\,048\,576$, corresponding to the values $k=2,\ldots,8$ in Eq.~(\ref{eqn:N}) with $q=4$ and $e=1$. Each entry in the table denotes the range of the number $n$ of qubits that can be achieved by a BOA scheme of the corresponding locality and length. For instance, the entry $7$--$17$ at location $(3,2\,048)$ indicates that in order to decouple a $3$-local Hamiltonian on a system with $n$ qubits, where $n \in \{7, \ldots, 17\}$, the best known BOA schemes have $2\,048$ time steps. If the number of qubits is one higher, e.\,g., $n=18$, then the currently best known BOA scheme would require $10\,240$ time steps. Superscripts indicate if the dual codes $[n,k^\perp,\dist^\perp]_4$ underlying the BOAs were obtained by a particular construction: 
a) all codes for $\ell=2$ were obtained from the Hamming code family $[n, n-k, 3]_4$ with $k = \log_4(3n+1)$; 
b) the code $[6,3,4]_4$ is the Hexacode \cite{MS:77}; 
c) the codes with parameters $[17,13 ,4]_4$,  $[41,36,4]_4$, $[126,120,4]_4$, $[288,281,4]_4$, and $[756,748,4]_4$ are based on caps in finite projective spaces which are sets of points of which no three are collinear, see \cite{EB:2001}; 
d) the code $[11,6,5]_4$ is a quadratic residue code, see \cite[5.13]{HSS99} and \cite{MS:77}; 
e) the code $[21,15,5]_4$ is the Kschischang-Pasupathy code, see \cite{KP:92}; and
f) the code $[12,6,6]_4$ is a quadratic residue code, see \cite[5.13]{HSS99} and \cite{MS:77}. All other codes in the table are based on the database of best known linear codes that is available in Magma \cite{Magma}.}
\end{table}

\begin{table}
\begin{tabular}{c@{\quad}@{\quad}c@{\qquad}c@{\qquad}c@{\qquad}c@{\qquad}c@{\qquad}c}
\hline
$\ell \; \backslash \; N$  &  $324$  &  $4\,374$  &  $52\,488$      &  
590\,490      &  6\,377\,292   &  66\,961\,566 \\
\hline
2 & 2--$10^a$& 11--$91^a$& 92--$820^a$& 821--$7\,381^a$ & 7\,382--$66\,430^a$ 		  & 66\,431--$597\,871^a$ \\
3 & - 	 & 3--$10^b$ & 11--$82^b$ & 83--$212^b$ & 213--$840^b$ & 841--$6\,723^b$ \\
4 & - 	 & - 	 & 4--10 	& 11--20  & 21--72   & 73--96  \\
5 & - 	 & - 	 & - 	    & 5--10  & 11--16   & 17--73  \\
6 & - 	 & - 	 & - 	  & - 		& 6--10 & 11--17	 \\
7 & - 	 & - 	 & - 	  & - 		& - 	  & 7--10  \\
\hline
\end{tabular}
\caption{\label{tab:gf9} Table of the best known balanced-cycle orthogonal arrays (BOAs) for qutrit ($d=3$) systems, indicating the number of qutrits that can be decoupled by a BOA scheme for the given locality and length. Shown are the locality $\ell$ of the underlying Hamiltonian from $2$ up to $7$ and length $N=9^k 2k$ of the BOA cycles from $324$ up to $66\,961\,566$, corresponding to the values $k=2,\ldots,7$ in Eq.~(\ref{eqn:N}) with $q=9$ and $e=1$. Each entry in the table denotes the range of the number $n$ of qutrits that can be achieved by a BOA scheme of the corresponding locality and length. Superscripts indicate if the dual codes $[n,k^\perp,\dist^\perp]_9$ underlying the BOAs were obtained by a particular construction: 
a) all codes for $\ell=2$ were obtained from the Hamming code family $[n, n-k, 3]_9$ with $k = \log_9(8n+1)$; and 
b) the codes with parameters $[10,7,4]_9$,  $[82,78,4]_9$, $[212,207,4]_9$, $[840,834,4]_9$, and $[6\,723,6\,716,4]_9$ are based on caps in finite projective spaces, see \cite{EB:2001}. All other codes in the table are based on the database of best known linear codes that is available in Magma \cite{Magma}.}
\end{table}

\section{Examples}\label{sec:examples}

\begin{example}[2-local decoupling of a diagonal Hamiltonian]\label{ex:pedagogical}
We first consider a simple case of decoupling a 2-local Hamiltonian on 7 qubits, where we assume (to simplify the example) that the Hamiltonian is diagonal, i.e. consists only of Pauli $Z$ operators. In this case, it turns out that we need not use an irreducible representation and can consequently use $q=d=2$ (instead of $q=d^2=4$); to avoid clutter, we defer the proof that this works to Example~\ref{ex:diagonal} where we will consider a similar situation. Because $q=2$, we use the group $\F_2 = \Z_2 = \lbrace{0,1}\rbrace$ and choose the representation
\[
\rho : \Z_2 \rightarrow \lbrace \id,X \rbrace, 
\quad
\text{ with }
\quad
\rho(0) = \id,
\quad
\rho(1) = X
\]
and corresponding control unitaries
\begin{equation}\label{eqn:ped_example_unitary}
u_{0}(\delta) = \id, 
\quad
u_{1}(\delta) = e^{-i X \delta},
\quad
\text{over time } \delta \in [0,\tfrac{\pi}{2}].
\end{equation}
Note that by evolving over time $\Delta = \tfrac{\pi}{2}$, we can therefore implement (up to phase) $u_0(\Delta) = \id = \rho(0)$ and $u_1(\Delta) = X = \rho(1)$. We assume that we can perform these control unitaries on any qubit.

With our constraints of 7 qubits ($n=7$) with 2-local interactions ($\loc=2$) and our ability to use $q=2$, we seek an $[n,k]_q = [7,k]_2$ code with dual distance $\dist^\perp=3 = \loc+1$ for some (hopefully small) dimension $k$. We find that there is a $[7,3]_2$ code with this desired dual distance, given by the generator matrix 
\[
G = \left[
\begin{array}{ccc}
1 &	0 &	0 \\
0 &	1 &	0 \\
1 &	1 &	0 \\
0 &	0 &	1 \\
1 &	0 &	1 \\
0 &	1 &	1 \\
1 &	1 &	1 \\
\end{array}
\right].
\]
Observe that this code has dimension $k=3$, which will dictate the efficiency of the protocol.

Although irrelevant for our concerns here, one may observe that, as guaranteed by Theorem~\ref{theorem:OA_codes}, an array built from the codewords of this code is an orthogonal array; indeed, the OA shown in Fig.~\ref{fig:OA} was constructed from this code.
We, on the other hand, wish to create a BOA from this code. As per Theorem~\ref{theorem:BOA_exist}, we start with the (additive) group $\F_q^k = \Z_2^3$ and choose the generating set $\Sqk = S_2^{(3)} = \left\lbrace {\MiniMatrix{1\\0\\0}}, {\MiniMatrix{0\\1\\0}}, {\MiniMatrix{0\\0\\1}} \right\rbrace$. We set $\calL$ to be the Eulerian cycle on the Cayley graph ${\small \Gamma\Big(\Z_2^3,S_2^{(3)}\Big)}$ shown in Fig.~\ref{fig:cycle_cube}, namely
\begin{align*}
\calL_{\F_2^3} &= \Bigg(\!
\MiniMatrix{0\\0\\0},\! \MiniMatrix{1\\0\\0},\! \MiniMatrix{1\\1\\0},\! \MiniMatrix{0\\1\\0},
\MiniMatrix{0\\0\\0},\! \MiniMatrix{0\\1\\0},\! \MiniMatrix{1\\1\\0},\! \MiniMatrix{1\\0\\0},
\MiniMatrix{0\\0\\0},\! \MiniMatrix{0\\0\\1},\! \MiniMatrix{1\\0\\1},\!  \\& \MiniMatrix{1\\0\\0}, 
\MiniMatrix{1\\0\\1},\! \MiniMatrix{1\\1\\1},\! \MiniMatrix{1\\1\\0},\! \MiniMatrix{1\\1\\1},
\MiniMatrix{0\\1\\1},\! \MiniMatrix{0\\1\\0},\! \MiniMatrix{0\\1\\1},\! \MiniMatrix{0\\0\\1},
\MiniMatrix{0\\1\\1},\! \MiniMatrix{1\\1\\1},\! \MiniMatrix{1\\0\\1},\! \MiniMatrix{0\\0\\1} 
\!\Bigg).
\end{align*}

We map this cycle under the action of the generator matrix $G$ to obtain the array $A_{G\calL} = [Gm]_{m\in\calL_{\F_2^3}}$, which is precisely the array that was shown in Fig.~\ref{fig:BOA_highlight}. According to Theorem~\ref{theorem:BOA_exist}, this is a $BOA(N,n,q,\loc) = BOA(24,7,2,2)$ with $N = q^k|S_2^{(3)}| =  2^3 \cdot 3 = 24$. By definition, this means that every $\loc\times N$ subarray (obtained by only considering a selection of just $\loc$ of the $n$ rows) defines a balanced cycle on the Cayley graph $\Gamma(\curG,\curS)$ of $\curG=\F_q^\loc = \Z_2^2$ with respect to some generating set $\curS$ (which may depend on the subarray). For example, look at rows 5 and 7, highlighted in Fig.~\ref{fig:BOA_highlight}. This defines the balanced cycle on $\Gamma\Big(\Z_2^2, \left\lbrace {\MiniMatrix{0\\1}}, {\MiniMatrix{1\\1}}\right\rbrace\Big)$ that was shown in Fig.~\ref{fig:balanced_cycle_hourglass}. The generating set and balanced cycle depend on the choice of rows, but by virtue of being a BOA, some balanced cycle will be obtained for any choice of 2 rows.

According to Theorem~\ref{theorem:BOA}, the protocol of Protocol~\ref{prot:BOA} defined by this BOA performs bounded-strength decoupling on our 7-qubit 2-local system. To construct this protocol, we consider the transitions between the columns of the BOA, defining the schedule shown in Fig.~\ref{fig:BOA_schedule}. 
\begin{figure}
	\begin{center}
		\includegraphics[width=5in]{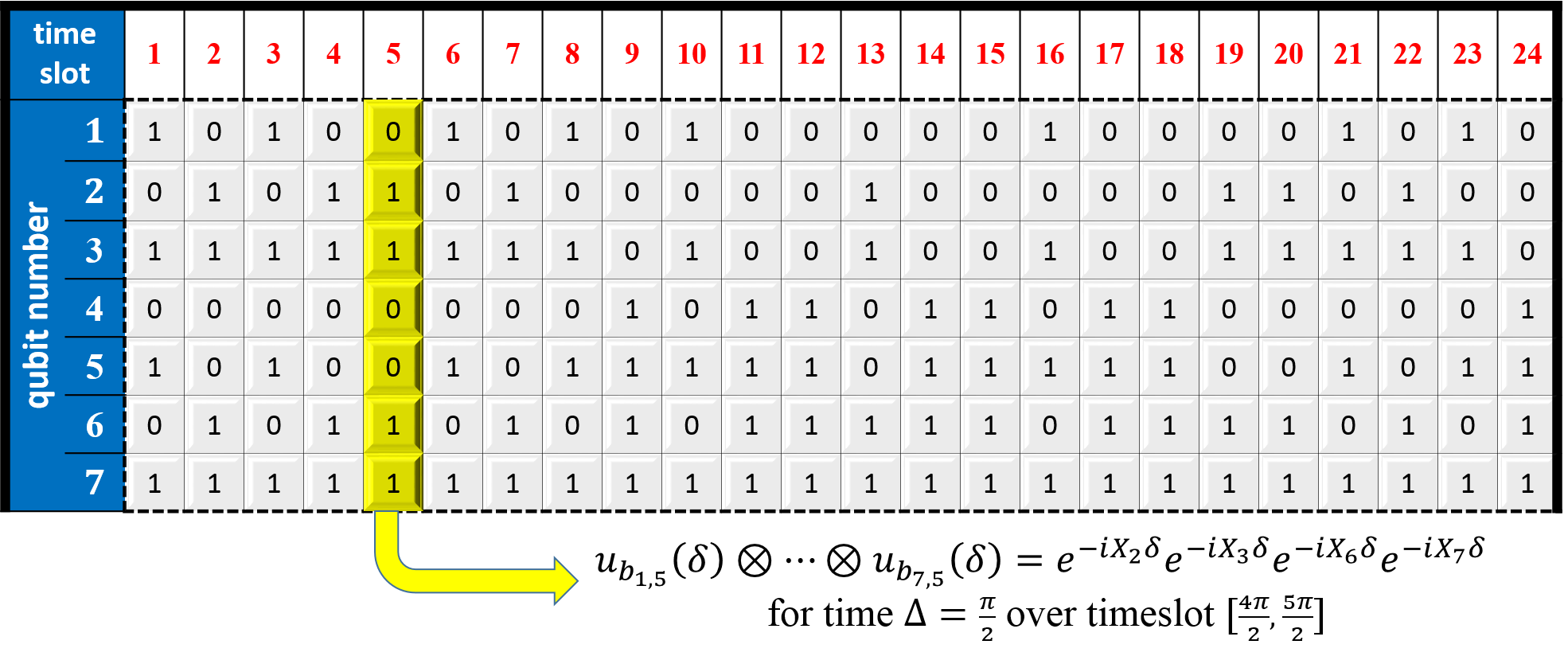}
		\caption[BOA decoupling schedule]
		{\label{fig:BOA_schedule}
		A $7\times 24$ array defining the decoupling protocol in Example~\ref{ex:pedagogical} in the format of Fig.~\ref{fig:array}. Rows correspond to qubit numbers, columns correspond to time slots (each of width $\Delta=\tfrac{\pi}{2}$), and entries correspond to unitary operators on qubits according to Eq.~\ref{eqn:ped_example_unitary}. 
As per Protocol~\ref{prot:BOA}, the control cycle evolution is $U_c\Big((j-1)\Delta+\delta\Big) = u_{\vec b_j}(\delta) U_c\Big((j-1)\Delta\Big)$, $\delta\in[0,\Delta]$, where $\vec b_j$ is the $j$th  column.		
For example, because $\vec b_5 = (0,1,1,0,0,1,1)^T$, we have 
$u_{\vec b_5}(\delta)
= u_0 \otimes u_1 \otimes u_1 \otimes u_0 \otimes u_0 \otimes u_1 \otimes u_1 \, (\delta) = e^{-iX_2\delta}e^{-iX_3\delta}e^{-iX_6\delta}e^{-iX_7\delta}$.
		}
	\end{center}
\end{figure}
The control unitaries to be applied are defined by these transitions and our choice of Eq.~\ref{eqn:ped_example_unitary}, which was chosen to be consistent with our representation $\rho$. For example, in time slot 5, the transition column is ${\vec b_5} = (0,1,1,0,0,1,1)^T$, which corresponds to the unitary  
\[
u_{\vec b_5}(\delta)= e^{-iX_2\delta}e^{-iX_3\delta}e^{-iX_6\delta}e^{-iX_7\delta}
\]
where $X_i$ denotes the Pauli $X$ operator on the $i$th qubit.
As per Protocol~\ref{prot:BOA}, the control cycle evolution is ${U_c\Big((j-1)\Delta+\delta\Big)} = {u_{\vec b_j}(\delta) U_c\Big((j-1)\Delta\Big)}$, $\delta\in[0,\Delta]$, where $\vec b_j$ is the $j$th  column in Fig.~\ref{fig:BOA_schedule}.		
This protocol will decouple any 2-local 7-qubit diagonal Hamiltonian.

We epmhasize that in this simple diagonal-Hamiltonian example, we were able to use $q=d=2$ (for reasons that will be addressed in Example~\ref{ex:diagonal}). If the Hamiltonian were not known to be diagonal, this would not in general have been possible, and we would have needed to instead use a $[7,k]_q$ code for $q=4$.

\end{example}

\begin{example}[2-local decoupling using a Hamming code]
Consider an arbitrary $2$-local Hamiltonian $H$ on a system of $5$ qubits. Then $H$ can be decoupled by applying a BOA derived from the code dual to a $[5,3,3]_4$ Hamming code, namely the $[5,2]_4$ code over $\F_4$ with the generator matrix 
\[
G = \left[
\begin{array}{cc}
	1 & 0 \\
	0 & 1 \\
	1 & \alpha^2 \\
	\alpha^2 & \alpha^2 \\
	\alpha^2 & 1\\
\end{array}
\right],
\]
where $\alpha$ is a primitive element of order $3$ of $\F_4$. 
Note that we arrange the code words as {\em column} vectors, consistent with the notation used throughout this paper and some -- but not all -- of the literature. Since here $k=2$, $d=2$, and $e=1$, the corresponding BOA has a total number of time steps given by $N=d^{2k} 2ke=64$. When arranged into the columns of a $5 \times 64$ matrix, each of the $64$ control Hamiltonians that are applied in this scheme corresponds to one of the $16$ code words of the $[5,2]_4$ code.  
\end{example}

\newcommand\nix[1]{{}}

\nix{
\begin{example}
Consider an arbitrary $4$-local Hamiltonian $H$ on a system of $16$ qubits. Then $H$ can be decoupled by applying a BOA derived from the dual code of a ${\rm BCH^{ext}}(\F_{16}/\F_4,5) = [16,9,5]_4$, i.e. from a code over $\F_4$ with parameters $[16,7]_4$ and generator matrix
\[
G = \left[
\begin{array}{cccccccccccccccc}
1& 0& 0& 0& 0& 0& 0& 1& \alpha^2& 1& 1& \alpha& \alpha^2& \alpha^2& \alpha& \alpha^2\\
0& 1& 0& 0& 0& 0& 0& \alpha& 0& 1& \alpha^2& \alpha& \alpha^2& \alpha& 0& \alpha \\
0& 0& 1& 0& 0& 0& 0& 0& \alpha& 0& 1& \alpha^2& \alpha& \alpha^2& \alpha& \alpha \\
0& 0& 0& 1& 0& 0& 0& \alpha& 1& 0& \alpha& \alpha& \alpha& \alpha^2& 0& \alpha^2 \\
0& 0& 0& 0& 1& 0& 0& 0& \alpha& 1& 0& \alpha& \alpha& \alpha& \alpha^2& \alpha^2 \\
0& 0& 0& 0& 0& 1& 0& \alpha^2& \alpha& 1& \alpha& 1& 0& 0& \alpha^2& 1 \\
0& 0& 0& 0& 0& 0& 1& \alpha^2& 1& 1& \alpha& \alpha^2& \alpha^2& \alpha& 1& \alpha^2
\end{array}
\right]^T,
\]
where $\alpha$ is a primitive element (of order $(16-1)/(4-1)=5$) for the extension $\F_{16}/\F_4$. Since $k=7$, $d=2$, and $e=1$, the corresponding BOA has a total number of time steps given by $N=d^{2k} 2ke=229,376$. When arranged into the columns of an $16 \times 229,376$ matrix, each of the $229,376$ control Hamiltonians that are applied in this scheme corresponds to one of the $16,384$ code words of the $[16,7]_4$ code. 
\end{example}
}

\begin{example}[5-local decoupling of a diagonal Hamiltonian using a BCH code]\label{ex:diagonal}
Recall from Remark~\ref{remark:reducible} that if one is interested in decoupling a Hamiltonian of a particular form, it may not be necessary for $\rho$ to be irreducible, and in such a case it may be possible to choose a code over a field $\F_q$ for which $q$ is less than $d^2$.
Consider a diagonal (i.e., $Z$-only) $5$-local Hamiltonian $H$ on a system of $16$ qubits. Then $H$ can be decoupled by applying a BOA derived from the dual code of a ${\rm BCH^{ext}}(\F_2^4/\F_2,6) = [16,7,6]_2$, i.e. from a code over $\F_2$ with parameters $[16,9]_2$ and generator matrix
\[
G = \left[
\begin{array}{ccccccccc}
1&	0&	0&	0&	0&	0&	0&	0&	0	\\
0&	1&	0&	0&	0&	0&	0&	0&	0	\\
0&	0&	1&	0&	0&	0&	0&	0&	0	\\
0&	0&	0&	1&	0&	0&	0&	0&	0	\\
0&	0&	0&	0&	1&	0&	0&	0&	0	\\
0&	0&	0&	0&	0&	1&	0&	0&	0	\\
0&	0&	0&	0&	0&	0&	1&	0&	0	\\
0&	0&	0&	0&	0&	0&	0&	1&	0	\\
0&	0&	0&	0&	0&	0&	0&	0&	1	\\
1&	1&	0&	0&	1&	1&	1&	0&	0	\\
0&	1&	1&	0&	0&	1&	1&	1&	0	\\
0&	0&	1&	1&	0&	0&	1&	1&	1	\\
1&	1&	0&	1&	0&	1&	1&	1&	1	\\
1&	0&	1&	0&	0&	1&	0&	1&	1	\\
1&	0&	0&	1&	1&	1&	0&	0&	1	\\
1&	0&	0&	0&	1&	0&	1&	1&	1	\\
\end{array}
\right].
\]
Due to the special structure of the Hamiltonian, we are able to choose $q=d=2$ (rather than $q=d^2=4$) in this case.
Since $k=9$, $d=2$, $e=1$, and the Hamiltonian is $Z$ only, the corresponding BOA has a total number of time steps given by $N=d^{k} ke=4\,608$. When arranged into the columns of an $16 \times 4\,608$ matrix, each of the $4\,608$ control Hamiltonians that are applied in this scheme corresponds to one of the $512$ code words of the $[16,9]_2$ code.  

To construct our protocol from this code we first choose a generating set $S_{2}^{(9)}$ for $\F_2^9$, such as the $k=9$ standard basis vectors $\lbrace (1,0,0,\ldots)^T, (0,1,0,0,\ldots)^T,\ldots \rbrace$. We then find an Eulerian cycle $\calL$ on the Cayley graph $\Gamma(\F_2^9,S_{2}^{(9)})$ and map it to an Eulerian cycle $G\calL$ on the Cayley graph $\Gamma(G[\F_2^9],G[S_{2}^{(9)}])$ using the generator matrix above. Our choice of $S_{2}^{(9)}$ as being the standard basis vectors would dictate that the transition labels $\vec b = Gs$, for $s\in S_{2}^{(9)}$, are simply the columns of $G$. Our BOA consists of the $2^9=512$ 16-bit code words, each appearing exactly 9 times according to the order specified by $G\calL$. To use the BOA as a decoupling scheme, we may choose 
\[
\rho : \lbrace 0,1 \rbrace \rightarrow \lbrace \id,X \rbrace, 
\quad
\text{ with }
\quad
\rho(0) = \id,
\quad
\rho(1) = X
\]
and choose the corresponding single-qubit control unitaries to be
\begin{equation*}
u_{0}(\delta) = \id, 
\quad
u_{1}(\delta) = e^{-i X \delta} ,
\quad
\text{over time } \delta \in [0,\tfrac{\pi}{2}].
\end{equation*}
Observe that (ignoring global phase) $u_0(\tfrac{\pi}{2}) = \id = \rho(0)$ and $u_1(\tfrac{\pi}{2}) = X = \rho(1)$. 
The multi-qubit control unitaries are defined by $u_{\vec b} = u_{b_1}\otimes\cdots\otimes u_{b_{16}}$ (for ${\vec b} \in \lbrace{0,1\rbrace}^{16}$).
For example, if $s=(1,0,0,\ldots)^T$, then $\vec b = Gs$ is the first column of $G$ and $u_{\vec b} = e^{-i X_1 \delta}e^{-i X_{10} \delta}e^{-i X_{13} \delta}e^{-i X_{14} \delta}e^{-i X_{15} \delta}e^{-i X_{16} \delta}$ acting non-trivially on qubits 1, 10, 13, 14, 15, and 16.
The control scheme in Protocol~\ref{prot:BOA} is thus specified.

We now prove that this example works, even though $\rho$ is reducible (i.e. even though we are choosing $q=2$ rather than $q=4$). As per the argument in the proof of Theorem~\ref{theorem:BOA}, we need only focus on a single 5-local term of $H$ (so assume without loss of generality that $H$ consists of only one such term), we can ignore all but the 5 qubits on which it acts non-trivially, and we need only speak of the 5-qubit unitaries 
$u_\curs(\delta)$ 
that act on those qubits and correspond to $\curs \in \curS$ (where $\curS$ is the generator set for $\F_2^5$ derived from the BOA for those 5 qubits). 
According to the proof of Theorem~\ref{theorem:inefficient-bounded-balanced}, our scheme works if and only if $\Pi_\curG\Big(F_\curS(H)\Big)=0$. Here, however, $\rho$ is not irreducible, so $\Pi_\curG$ will not suppress all traceless operators; indeed, $X$-only operators commute with each $U_\curg$ and are therefore unmodified by $\Pi_\curG$. 

To show that $\Pi_\curG\Big(F_\curS(H)\Big)=0$ nevertheless holds, observe from Eq.~(\ref{eqn:F}) that each term in $F_\curS(H)$ is of the form
$u_\curs^\dag H u_\curs$.
Now, $H$ is diagonal, i.e. a tensor product of only $\id$ and $Z$, and $u_\curs(\delta)$ is a tensor product of only $\id$ and $e^{-iX\delta}$. Therefore, because 
$e^{iX\delta} Z e^{-iX\delta} = \cos(2\delta) Z + \sin(2\delta) Y$,
we see that $u_\curs^\dag H u_\curs$ can be expanded as a sum of tensor products of $\id$, $Z$ and $Y$. 
Moreover, because $H$ is traceless and conjugation by a unitary is trace-preserving, this sum cannot contain a term proportional to the identity, $\id^{\otimes 5}$. Thus, each of these terms consists of at least one operator (which for notational purposes we take to be on the first qubit) that is a $Z$ or a $Y$, i.e. each can be written in the form $\sigma \otimes A$, where $\sigma \in \lbrace Y,Z \rbrace$ and $A$ is some 4-fold tensor product of operators from $\lbrace \id,Y,Z \rbrace$. 
Our protocol is defined by a BOA of strength 5, so any subset of 5 rows of the BOA consists of all $2^5$ 5-tuples in $\F_2^5$ repeated an equal number of times.
Thus the sum in $\Pi_\curG$ involves conjugating by each $U_\curg$ where $U_\curg$ ranges over all $2^5$ possible tensor products that can be formed on 5 qubits using $\id$ and $X$. Focusing on the first qubit, we can equivalently say that the $U_\curg$ range over all possible $\id \otimes B$ and $X \otimes B$, where $B$ ranges over $\lbrace B_2 \otimes B_3 \otimes B_4 \otimes B_5 : B_i \in \lbrace \id,X \rbrace \rbrace$.
Conjugating $\sigma \otimes A$ by $\id \otimes B$ yields either $\sigma \otimes A$ or $- \sigma \otimes A$, whereas conjugating instead by $X \otimes B$ yields the same result but with the opposite sign (since $\sigma \in \lbrace Y,Z \rbrace$). In other words, 
$
(\id \otimes B) (\sigma \otimes A) (\id \otimes B)
+
(X \otimes B) (\sigma \otimes A) (X \otimes B)
= 
0.
$
Thus, the sum in $\Pi_\curG$ cancels in pairs, i.e. $\Pi_\curG\Big(F_\curS(H)\Big)$ is indeed 0.
\end{example}

\section{Conclusion}
We have shown how to use bounded-strength controls to decouple $n$ interacting qudits of dimension $d=p^e$ (for some prime $p$ and positive integer $e$) that obey an $\loc$-local Hamiltonian. 
The system may be either closed or open (i.e. coupled to an environment), as long as both the system Hamiltonian and the environmental couplings are $\loc$-local on the system.
The decoupling scheme is described using a balanced-cycle orthogonal array, which we introduced and showed how to construct from classical linear codes.
To determine the best possible scheme based on our method, we have to find the best linear error-correcting code $C^\perp=[n,k^\perp]_q$ of length $n$ and distance at least $\loc+1$.  By \textit{the best}, we mean $k^\perp$ should be maximized for the given system size ($n$) and locality ($\loc$). 
The construction in the present paper yields a decoupling scheme that uses $N= d^{2k} 2ke$ time slices (of fixed length) where $k=n-k^\perp$.

Finding the best code is a key problem in the theory of error-correcting codes; extensive code tables have been compiled for small distances. For the important case of qubits with 2-local interactions, for example, one can use Hamming codes over $\F_4$ such that $k = \log_4(3n+1)$, whence $N$ scales like $n\log n$. For higher degrees of locality, we can use families of BCH codes to construct the decoupling schemes. The designed distance of these codes is chosen based on the locality $\loc$ of the Hamiltonian, leading to a scaling of $N$ as $n^{\loc -1} \log n$. 
An open question is whether the schemes so derived are optimal in the asymptotic sense, i.e. whether, for fixed $\loc$ and qudit dimension $d$, a better scaling with $n$ is possible. We note that it is known \cite{JWB02} that when using bang-bang pulses, time at least $\Omega(n)$ is necessary to decouple general 2-body Hamiltonians, whereas our bounded-strength scheme takes time $O(n \log n)$ using Hamming codes for such Hamiltonians.
Another interesting open question is to develop a theory for systems with mixed qudit dimensions. All schemes derived here are decoupling schemes up to first order, and while it is easy to extend this to second order using symmetry, it would be interesting to find schemes that also achieve decoupling to higher orders. Finally, we mention as an avenue for future research the application of the derived bounded-strength decoupling schemes for the purpose of Hamiltonian simulation.  

\subsection*{Acknowledgments}
This work was supported in part by the ARO grant Contract Number W911NF-12-0486, as well as 
by the National Science Foundation Science and Technology Center for Science of Information under grant CCF-0939370.  
P.W. gratefully acknowledges the support from the NSF CAREER Award CCF-0746600. We thank Madhu Sudan and Sergey Yekhanin for discussions on BCH codes, in particular regarding the proof of the bound in Theorem \ref{thm:bchbound}.

%
%

\bibliographystyle{adam}
\bibliography{BOArefs}

\end{document}